\documentclass[letterpaper,11pt]{article}
\usepackage{amsmath,amsthm,amsfonts,amssymb,bm}
\usepackage{mathtools}
\usepackage[ruled,vlined,linesnumbered]{algorithm2e}
\usepackage[margin=1in]{geometry}
\usepackage{nicefrac}
\usepackage{hyperref}
\usepackage{paralist}

\hypersetup{
  colorlinks   = true, %Colours links instead of ugly boxes
  urlcolor     = blue, %Colour for external hyperlinks
  linkcolor    = blue, %Colour of internal links
  citecolor   = red %Colour of citations
}

\newtheorem{theorem}{Theorem}[section]

\newtheorem{corollary}[theorem]{Corollary}
\newtheorem{lemma}[theorem]{Lemma}
\newtheorem{proposition}[theorem]{Proposition}
\newtheorem{observation}[theorem]{Observation}

\newcommand{\DeclareAutoPairedDelimiter}[3]{%
  \expandafter\DeclarePairedDelimiter\csname Auto\string#1\endcsname{#2}{#3}%
  \begingroup\edef\x{\endgroup
    \noexpand\DeclareRobustCommand{\noexpand#1}{%
      \expandafter\noexpand\csname Auto\string#1\endcsname*}}%
  \x}

\DeclareMathOperator{\polylog}{polylog}
\DeclareMathOperator{\poly}{poly}
\DeclareMathOperator{\supp}{supp}
\newcommand{\RFunc}{\mathtt{R}}
\newcommand{\groundset}{\mathcal{N}}
\newcommand{\realnum}{\mathbb{R}}
\newcommand{\nonnegative}{\realnum_{\geq 0}}
\newcommand{\cE}{\mathcal{E}}
\newcommand{\tx}{\tilde{x}}

\DeclareAutoPairedDelimiter\power{\lvert}{\rvert}
\DeclareAutoPairedDelimiter{\ceil}{\lceil}{\rceil}
\DeclareAutoPairedDelimiter{\bracks}{[}{]}
\DeclareAutoPairedDelimiter{\curly}{\{}{\}}

\newcommand{\range}[2]{\ensuremath{\bracks{#1,#2}}}
\newcommand{\cube}[3]{\ensuremath{\range{#1}{#2}^{#3}}}
\newcommand{\Exp}[1]{\ensuremath{\mathbb{E}\bracks{#1}}}
\newcommand{\pr}[1]{\ensuremath{\Pr\bracks{#1}}}

\newcommand{\Func}[3]{\ensuremath{#1\colon#2\rightarrow#3}}
\newcommand{\onevect}[1]{\ensuremath{\mathbf{1}_{#1}}}
\newcommand{\onenorm}[1]{\ensuremath{\left\|{#1}\right\|_1}}
\newcommand{\email}[1]{{\href{mailto:#1}{#1}}}

\makeatletter
\newcommand{\newreptheorem}[2]{%
	\newenvironment{rep#1}[1]{%
	\expandafter\renewcommand\csname the#2\endcsname{\ref*{##1}}%
	\expandafter\renewcommand\csname theH#2\endcsname{repeat.##1}%
	\begin{#1}}%
	{\end{#1}%
	\addtocounter{#2}{-1}}}
\makeatother

\newreptheorem{theorem}{theorem}

\begin{document}
\iffalse --------------------------------------------------- TITLE PAGE --------------------------------------------------- \fi
%\title{Streaming Submodular Maximization Subject to Cardinality Constraint}
\title{Making a Sieve Random: Improved Semi-Streaming Algorithm for Submodular Maximization under a Cardinality Constraint}
\author{
Naor Alaluf\thanks{Department of Mathematics and Computer Science, Open University of Israel. E-mail: \email{naoralaluf@gmail.com}}
\and
Moran Feldman\thanks{Department of Mathematics and Computer Science, Open University of Israel. E-mail: \email{moranfe@openu.ac.il}}
}
\maketitle
\begin{abstract}
In this paper we consider the problem of maximizing a non-negative submodular function subject to a cardinality constraint in the data stream model. Previously, the best known algorithm for this problem was a $5.828$-approximation semi-streaming algorithm based on a local search technique (Feldman et al., 2018). For the special case of this problem in which the objective function is also monotone, the state-of-the-art semi-streaming algorithm is an algorithm known as Sieve-Streaming, which is based on a different technique (Badanidiyuru, 2014). Adapting the technique of Sieve-Streaming to non-monotone objective functions has turned out to be a challenging task, which has so far prevented an improvement over the local search based $5.828$-approximation. In this work, we overcome the above challenge, and manage to adapt Sieve-Streaming to non-monotone objective functions by introducing a ``just right'' amount of randomness into it. Consequently, we get a semi-streaming polynomial time $4.282$-approximation algorithm for non-monotone objectives. Moreover, if one allows our algorithm to run in super-polynomial time, then its approximation ratio can be further improved to $3 + \varepsilon$.
\end{abstract}
\pagenumbering{Alph}
\thispagestyle{empty}
\clearpage
\setcounter{page}{1}
\pagenumbering{arabic}

\section{Introduction}
Submodular functions are a wide class of functions capturing the intuitive notion of diminishing returns. As diminishing returns occurs naturally in many scenarios, the optimization of submodular functions subject to combinatorial constraints has found many applications in diverse fields, including machine learning~\cite{DBLP:conf/icml/DasK11,DBLP:journals/jair/GolovinK11,DBLP:conf/miccai/SalehiKSC17}, social networks~\cite{DBLP:conf/www/HartlineMS08,DBLP:journals/toc/KempeKT15} and algorithmic game theory~\cite{DBLP:journals/toc/DughmiRS12,DBLP:journals/disopt/SchulzU13}.

In the context of many of the above applications, it is desirable for the submodular optimization algorithm to be a (semi-)streaming algorithm because the input is either very large (e.g., the friendships graph of a social network) or naturally occurs at the form of a long stream (e.g., summarizations of the frames generated by a surveillance camera). In response to this need, Badanidiyuru et al.~\cite{DBLP:conf/kdd/BadanidiyuruMKK14} and Chakrabarti and Kale~\cite{DBLP:journals/mp/ChakrabartiK15} developed two semi-streaming algorithms for the maximization of non-negative submodular functions that are also \emph{monotone}. A set function $f\colon 2^\groundset \to \realnum$ over a ground set $\groundset$ is submodular if for every two sets $A \subseteq B \subseteq \groundset$ and element $u \in \groundset \setminus B$ it holds that
\[
	f(A \cup \{u\}) - f(A)
	\geq
	f(B \cup \{u\}) - f(B)
	\enspace,
\]
and it is monotone if $f(A) \leq f(B)$ for every two sets $A \subseteq B \subseteq \groundset$. The algorithm of Badanidiyuru et al.~\cite{DBLP:conf/kdd/BadanidiyuruMKK14} maximizes non-negative monotone submodular functions subject to a cardinality constraint up to an approximation ratio of $2$; and the algorithm of Chakrabarti and Kale~\cite{DBLP:journals/mp/ChakrabartiK15} achieves a $4$-approximation for the maximization of the same kind of functions subject to a more general class of constraints known as matroid constraints.

Note again that the two above algorithms work only for monotone submodular functions. To handle non-monotone submodular functions, it is usually necessary to use randomness.\footnote{Some works use alternative techniques involving the maintenance of multiple solutions. However, it is often natural to view these solutions as the support of a distribution of solutions. See~\cite{DBLP:journals/talg/BuchbinderF18} for an explicit example of this point of view.} The algorithm of Chakrabarti and Kale~\cite{DBLP:journals/mp/ChakrabartiK15} is local-search based, and easily integrates with randomness, which has lead to multiple works adapting it to non-monotone functions~\cite{DBLP:conf/icalp/ChekuriGQ15,DBLP:conf/nips/FeldmanK018,DBLP:conf/aaai/MirzasoleimanJ018}. The best of these adaptations achieves an approximation ratio of $3 + 2\sqrt{2} \approx 5.828$ for maximizing a non-negative (not necessarily monotone) submodular function subject to a matroid constraint~\cite{DBLP:conf/nips/FeldmanK018}. In contrast, the algorithm of Badanidiyuru et al.~\cite{DBLP:conf/kdd/BadanidiyuruMKK14} for cardinality constrains, which is known as Sieve-Streaming, is based on a thresholding technique. In a nutshell, the algorithm picks a threshold and then selects every element whose marginal contribution to the current solution of the algorithm exceeds this threshold, until the solution gets to the maximum cardinality allowed. The analysis of Sieve-Streaming then handles in two very different ways the case in which the solution grew all the way to the maximum cardinality allowed and the case in which this did not happen. Unfortunately, most natural ways to add randomness to Sieve-Streaming result in input instances for which both these cases might happen with a non-zero, which makes the analysis break down. Due to this hurdle, prior to this work, no random adaptation of Sieve-Streaming managed to improve over the $5.828$-approximation of~\cite{DBLP:conf/nips/FeldmanK018} despite the significant advantage of Sieve-Streaming over the local search approach of  Chakrabarti and Kale~\cite{DBLP:journals/mp/ChakrabartiK15} in the context of monotone submodular functions.\footnote{Chekuri et al.~\cite{DBLP:conf/icalp/ChekuriGQ15} claimed an improved approximation ratio of $4.718$ for cardinality constraints based on such an adaptation of Sieve-Streaming, but an error was later found in the proof of this improved ratio~\cite{communication:Chekuri18}. See Appendix~\ref{app:error} for more detail.}

In this paper, we present a novel way to introduce randomness into the thresholding technique of Sieve-Streaming. Our basic idea is to base the decisions of the algorithm on the values of expectations over appropriately chosen distributions. On the one hand, this allows our algorithm to include the necessary random component, and on the other hand, since expectations have deterministic values, the algorithm we get is deterministic enough to allow us to consider at every given time only one of the cases from the analysis of Sieve-Streaming. Using this idea, we get the following theorem. In this theorem, and in the rest of the paper, we denote by $k$ the maximum number of elements allowed by the cardinality constraint in a solution. We also remind the reader that a \emph{semi-streaming} algorithm is an algorithm that gets its input in the form of a data stream and uses a memory whose size is bounded by the maximum size of a feasible solution up to a poly-logarithmic factor, which in our context means $O(k \cdot \polylog (k, |\groundset|))$ space.

\begin{theorem}
\label{trm:non-polynomial}
	For every constant $\varepsilon \in (0, 1]$, there exists a semi-streaming $(3 +\varepsilon)$-approximation algorithm for maximizing a non-negative submodular function
	subject to cardinality constraint. The algorithm stores at most $O(k\varepsilon^{-2} \log k)$ elements.
\end{theorem}

The time complexity of the algorithm whose existence is guaranteed by Theorem~\ref{trm:non-polynomial} is of the form $h(k, \varepsilon) \cdot \poly(|\groundset|)$ for some non-polynomial function $h$, which makes this algorithm useful in practice only when $k$ and $\varepsilon^{-1}$ are small enough. For larger values of $k$ and $\varepsilon^{-1}$, the guarantee of Theorem~\ref{trm:non-polynomial} is interesting only from an information theoretic point of view. To counteract this, the following theorem describes a different version of our algorithm that is more appropriate for practical uses. This version achieves a polynomial time complexity at the cost of guaranteeing a weaker approximation ratio.

\begin{theorem}
\label{trm:polynomial}
	There exists a \emph{polynomial time} semi-streaming $4.282$-approximation algorithm for maximizing a non-negative submodular function
	subject to cardinality constraint. This algorithm stores at most $O(k \log k)$ elements.
\end{theorem}

The algorithm guaranteed by Theorem~\ref{trm:polynomial} uses as a black box an offline algorithm for the problem of maximizing a non-negative submodular function subject to cardinality constraint, and naturally, its approximation ratio depends on the approximation ratio of this offline algorithm. The approximation ratio given by Theorem~\ref{trm:polynomial} was calculated based on the guarantee of the offline algorithm suggested by~\cite{DBLP:conf/soda/BuchbinderFNS14}, and will improve if better offline algorithms are found for the problem. However, even based on the existing offline algorithm, Theorem~\ref{trm:polynomial} significantly improves over the state-of-the-art $5.828$-approximation.

It is also worth mentioning that a simple adaptation of a result due to Buchbinder et al.~\cite{to_appear:BuchbinderFS19} shows that no data stream algorithm can achieve $(2 - \varepsilon)$-approximation (for any positive $\varepsilon$) for the problem we consider, unless it uses $\Omega(|\groundset|)$ memory. Moreover, this is true even if the algorithm is allowed to use unbounded computational power. For completeness, we include the proof of this result in Appendix~\ref{app:inapproximability}.

\subsection{Additional Related Work}

The problem of maximizing a non-negative monotone submodular function subject to a cardinality or a matroid constraint was studied (in the offline model) already in the $1970$'s. In $1978$, Nemhauser et al.~\cite{DBLP:journals/mp/NemhauserWF78} and Fisher et al.~\cite{Fisher1978} showed that a natural greedy algorithm achieves an approximation ratio of $\nicefrac{e}{e - 1} \approx 1.58$ for this problem when the constraint is a cardinality constraint and an approximation ratio of $2$ for matroid constraints. The $\nicefrac{e}{e - 1}$ approximation ratio for cardinality constraints was shown to be optimal already on the same year by Nemhauser and Wolsey~\cite{DBLP:journals/mor/NemhauserW78}, but the best possible approximation ratio for matroid constraints was open for a long time. Only a decade ago, Calinescu et al.~\cite{DBLP:journals/siamcomp/CalinescuCPV11} managed to show that a more involved algorithm, known as ``continuous greedy'', can achieve $\nicefrac{e}{e - 1}$-approximation for this type of constraints, which is tight since matroid constriants generalize cardinality constraints. Unlike the natural greedy algorithm, continuous greedy is a randomized algorithm, which raised an interesting question regarding the best possible approximation ratio for matroid constraints that can be achieved by a deterministic algorithm. Very recently, Buchbinder et al.~\cite{DBLP:conf/soda/BuchbinderF019} made a slight step towards answering this question. Specifically, they described a deterministic algorithm for maximizing a monotone submodular function subject to a matroid constraint whose approximation ratio is $1.997$. This algorithm shows that the $2$ approximation of the greedy algorithm is not the right answer for the above mentioned question.

Many works have studied also the offline problem of maximizing a non-negative (not necessarily monotone) submodular function subject to a cardinality or a matroid constraint~\cite{to_appear:BuchbinderF19,DBLP:conf/soda/BuchbinderFNS14,DBLP:journals/siamcomp/ChekuriVZ14,DBLP:conf/focs/EneN16,DBLP:journals/talg/Feldman17,DBLP:conf/focs/FeldmanNS11}. The most recent of these works achieves an approximation ratio of $2.598$ for both cardinality and matroid constraints~\cite{to_appear:BuchbinderF19}. In contrast, it is known that no polynomial time algorithm can achieve an approximation ratio of $2.037$ for cardinality constraints or $2.093$ for matroid constraints, respectively~\cite{DBLP:conf/soda/GharanV11}.

The study of data stream algorithms for submodular maximization problems is related to the study of online algorithms for such problems. A partial list of works on algorithms of the last kind includes~\cite{DBLP:conf/esa/AzarGR11,to_appear:BuchbinderFG19,to_appear:BuchbinderFS19,DBLP:journals/talg/ChanHJKT18,DBLP:conf/soda/KapralovPV13,DBLP:journals/siamcomp/KorulaMZ18}.
\section{Preliminaries}
In this section we introduce two standard extensions of submodular functions. The first of these extensions is known as the \textit{multilinear extension}. To define this extension, we first need to define the random set $\RFunc(x)$. For every vector $x \in \cube{0}{1}{\groundset}$, $\RFunc(x)$ is defined as a random subset of $\groundset$ that includes every element $u \in \groundset$ with probability $x_u$, independently. The multilinear extension $F$ of $f$ is now defined for every vector $x\in\cube{0}{1}{\groundset}$ by
\begin{align*}
	F(x)=\Exp{f\big(\RFunc(x)\big)}
	=\sum_{A\subseteq\groundset}{f(A)\cdot\pr{\RFunc(x)=A}}
	=\sum_{A\subseteq\groundset}{\left(f(A)\cdot\prod_{u\in{A}}{x_u}\cdot\prod_{u\notin{A}}{(1-x_u)}\right)}\enspace.
\end{align*}
One can observe from the definition that $F$ is indeed a multilinear function of the coordinates of $x$, as suggested by its name.

In the analysis of our algorithm we need an upper bound on the possible increase in the value of $F(x)$ when some of the indices of $x$ are zeroed. Corollary~\ref{cor:discardElements} provides such an upper bound. To prove it, we first need the following known lemma by Buchbinder et al.~\cite{DBLP:conf/soda/BuchbinderFNS14}.
\begin{lemma}[Lemma 2.2 from~\cite{DBLP:conf/soda/BuchbinderFNS14}]
\label{lem:atMostP}
	Let $\Func{f}{2^{\groundset}}{\nonnegative}$ be a non-negative submodular function. Denote by $A(p)$ a random subset of $A$
	where each element appears with probability at most $p$ (not necessarily independently). Then,
	$\Exp{f(A(p))}\geq(1-p) \cdot f(\varnothing)$.
\end{lemma}

In the statement of Corollary~\ref{cor:discardElements}, and in the rest of the paper, we denote by $\supp(x)$ the support of vector $x$, i.e., the set $\curly{u\in\groundset\mid x_u>0}$.

\begin{corollary}
\label{cor:discardElements}
	Let $\Func{f}{2^{\groundset}}{\nonnegative}$ be a non-negative submodular function whose multilinear extension is $F$, let $p$ be a number in the range $\range{0}{1}$ and let $x,y\in\cube{0}{1}{\groundset}$ be two vectors such that
\begin{compactitem}
	\item $\supp(x) \cap \supp(y) = \varnothing$,
	\item and	$y_u\leq p$ for every $u\in\groundset$.
\end{compactitem}
Then, $F(x+y)\geq(1-p)\cdot F(x)$.
\end{corollary}
\begin{proof}
	Let us define the function $G_x(S)=\Exp{f(\RFunc(x)\cup S)}$. It is not difficult to verify that $G_x$ is non-negative and submodular, and that $G_x(\varnothing)=F(x)$. Additionally, since $\supp(x) \cap \supp(y) = \varnothing$, $\RFunc(x+y)$ has the same distribution as $\RFunc(x)\cup\RFunc(y)$, and therefore,
	\begin{align*}
		F(x+y)={}&\Exp{f\big(\RFunc(x + y))} = \Exp{f\big(\RFunc(x)\cup\RFunc(y)\big)}\\={}&\Exp{G_x(\RFunc(y))}\geq(1-p)\cdot{G_x(\varnothing)}=(1-p)\cdot{F(x)}\enspace,
	\end{align*}
	where the inequality follows from Lemma~\ref{lem:atMostP}.
\end{proof}

The other standard extension of submodular functions that we need is the \textit{Lov\'{a}sz extension}. The Lov\'{a}sz extension $\hat{f}$ of $f$ is defined for every vector $x\in\cube{0}{1}{\groundset}$ by
\begin{align*}
	\hat{f} (x) = \int_0^1{f\left(T_{\lambda}(x)\right)d\lambda}\enspace,
\end{align*}
where $T_{\lambda}(x)=\curly{u\in\groundset\mid{x_u}\geq\lambda}$
is the subset of $\groundset$ containing every element $u \in \groundset$ whose corresponding coordinate $x_u$ in the vector $x$
is at least $\lambda$. In this paper we use the Lov\'{a}sz extension only to lower bound values of the multilinear extension via the following known lemma.
\begin{lemma}[Lemma~A.4 from~\cite{DBLP:journals/siamcomp/Vondrak13}] \label{lem:lovasz-multilinear}
Let $F$ and $\hat{f}$ be the multilinear and Lov\'{a}sz extensions of a submodular function $f\colon 2^\groundset \to \realnum$, respectively. Then, $F(x)\geq\hat{f}(x)$ for every vector $x\in\cube{0}{1}{\groundset}$.
\end{lemma}

We conclude this section by describing some additional notation that we use in the rest of the paper.  Given a set $A \subseteq \groundset$ and element $u \in \groundset$, we denote by \onevect{A} and \onevect{u} the characteristic vectors of $A$ and $\curly{u}$, respectively. Additionally, given two vectors $x,y \in [0, 1]^\groundset$, we use $x \vee y$ and $x \wedge y$ to denote the coordinate-wise maximum and minimum of $x$ and $y$, respectively (i.e., for every element $u \in \groundset$, $(x \vee y)_u = \max\{x_u, y_u\}$ and $(x \wedge y)_u = \min\{x_u, y_u\}$). We also use the shorthand $\partial_uF(x)$ for the first partial derivative
$\frac{\partial{F(x)}}{\partial{x_u}}$ of the multilinear extension $F$. Note that, since $F$ is multilinear,
\[
	\partial_uF(x) = F(x\vee\onevect{u})-F\big(x\wedge\onevect{\groundset\setminus\curly{u}}\big)
	\enspace.
\]
Finally, we denote by $OPT$ an arbitrary optimal solution for our problem, i.e., a subset of $\groundset$ of size at most $k$ maximizing $f$ among all such subsets.
\section{Our Algorithm} \label{sec:algorithm}
In this section, we introduce a simplified version of the algorithm we use to prove Theorems~\ref{trm:non-polynomial} and~\ref{trm:polynomial}. This simplified version (given as Algorithm~\ref{alg:main_algorithm}) captures our main new ideas, but avoids some technical issues that can be solved using standard techniques. In particular, Algorithm~\ref{alg:main_algorithm} assumes access to an estimate $\tau$ of $f(OPT)$ obeying $(1 - O(\varepsilon)) \cdot f(OPT) \leq \tau \leq f(OPT)$. Such an estimate can be produced, at the cost of increasing the space complexity of the algorithm by a factor of $O(\varepsilon^{-1} \log k)$, using a technique introduced by~\cite{DBLP:conf/kdd/BadanidiyuruMKK14}, and we defer the details to Appendix~\ref{apx:estimatingTau}.

Algorithm~\ref{alg:main_algorithm} gets two constant parameters $p \in (0, 1)$ and $c > 0$. The algorithm maintains a fractional solution $x \in [0, 1]^\groundset$. This fractional solution starts empty, and the algorithm adds to it fractions of elements as they arrive. Specifically, when an element $u$ arrives, the algorithm considers its marginal contribution with respect to the current fractional solution $x$. If this marginal contribution exceeds the threshold of $c\tau / k$, then the algorithms tries to add to $x$ a $p$-fraction of $u$, but might end up adding a smaller fraction of $u$ if adding a full $p$-fraction of $u$ to $x$ will make $x$ an infeasible solution, i.e., make $\|x\|_1 > k$ (note that $\|x\|_1$ is the sum of the coordinates of $x$).

After viewing all the elements, Algorithm~\ref{alg:main_algorithm} uses the fractional solution $x$ to generate two sets $S_1$ and $S_2$ that are feasible (integral) solutions. The set $S_1$ is generated by rounding the fractional solution $x$. Two rounding procedures, named Pipage Rounding and Swap Rounding, were suggested for this task in the literature~\cite{DBLP:journals/siamcomp/CalinescuCPV11,DBLP:conf/focs/ChekuriVZ10}. Both procedures run in polynomial time and guarantee that the output set $S_1$ of the rounding is always feasible, and that its expected value with respect to $f$ is at least the value $F(x)$ of the fractional solution $x$. The set $S_2$ is generated by finding a subset of the support of the vector $x$ that (approximately) maximizes $f$ among all such subsets of size at most $k$. One can observe that to get $S_2$ one must (approximately) solve the offline version of the problem of maximizing a non-negative submodular function subject to a cardinality constraint. In the pseudocode of Algorithm~\ref{alg:main_algorithm} we denote by $\alpha$ the inverse of the approximation ratio of the algorithm used to solve this problem and produce $S_2$. The value that we can assume for $\alpha$ depends on whether we want Algorithm~\ref{alg:main_algorithm} to run in polynomial time. If that is not required, then $\alpha$ can be assumed to be $1$ because the offline problem can be optimally solved using a brute force search. In contrast, $\alpha$ must have a smaller value if Algorithm~\ref{alg:main_algorithm} should run in polynomial time, and we discuss at a later point the value that can be assumed for $\alpha$ in this case. After computing the two feasible solutions $S_1$ and $S_2$, Algorithm~\ref{alg:main_algorithm} simply returns the better one of them.

\IncMargin{1em}
\begin{algorithm}[ht]
\label{alg:main_algorithm}
\caption{\texttt{Multilinear Threshold} $(p,c)$}
	\DontPrintSemicolon
	$x\gets\onevect{\varnothing}$.\;
	\For {each arriving element $u$} {
		\lIf {$\partial_uF(x)\geq\frac{c\tau}{k}$} {
			$x\gets{x}+\min\curly{p,k-\onenorm{x}}\cdot\onevect{u}$.
		}
	}
	Round the vector $x$ to yield a feasible solution $S_1$ such that $\Exp{f(S_1)} \geq {F(x)}$.\\
	Find a feasible solution $S_2\subseteq\supp(x)$ such that $\Exp{f(S_2)} \geq \alpha \cdot \arg\max_{S \subseteq \supp(x), |S| \leq k} f(S)$.\\
	\Return the better solution among $S_1$ and $S_2$.
\end{algorithm}\DecMargin{1em}

Let us denote by $\hat{x}$ the final value of the fractional solution $x$ (i.e., its value when the stream ends). We begin the analysis of Algorithm~\ref{alg:main_algorithm} with the following useful observation.
\begin{observation} \label{obs:x_type}
If $\|\hat{x}\|_1 < k$, then $\hat{x}_u = p$ for every $u \in \supp(\hat{x})$. Otherwise, this is still true for every element $u \in \supp(\hat{x})$ except for maybe a single element. 
\end{observation}
\begin{proof}
For every element $u$ added to the support of $x$ by Algorithm~\ref{alg:main_algorithm}, the algorithm sets $x_u$ to $p$ unless this will make $\|x\|_1$ exceed $k$, in which case the algorithm set $x_u$ to be the value that will make $\|x\|_1$ equal to $k$. Thus, after a single coordinate of $x$ is set to a value other than $p$ (or the initial $0$), $\|x\|_1$ becomes $k$ and the Algorithm~\ref{alg:main_algorithm} stops changing $x$.
\end{proof}

Using the last observation we can now bound the space complexity of Algorithm~\ref{alg:main_algorithm}, and show (in particular) that it is a semi-streaming algorithm for a constant $p$.
\begin{observation}
Assuming it takes $O(1)$ space to store an element of $\groundset$ and a value returned by $F$, Algorithm~\ref{alg:main_algorithm} can be implemented so that it stores at most $O(k/p)$ elements and its space complexity is $\tilde{O}(k/p)$, excluding the space complexity required by the algorithm for computing $S_2$.
\end{observation}
\begin{proof}
To calculate the sets $S_1$ and $S_2$, Algorithm~\ref{alg:main_algorithm} needs access only to the elements of $\groundset$ that appear in the support of $x$. Thus, the number of elements it needs to store is $O(|\supp(x)|) = O(k/p)$, where the equality follows from Observation~\ref{obs:x_type}.

Since each one of the sets $S_1$ and $S_2$ contains at most $k$ elements, they require $O(k)$ space. In addition, Algorithm~\ref{alg:main_algorithm} needs to store the vector $x$. It is possible to store the coordinates of $x$ taking the value of $p$ by storing their indices, which requires $O(\log |\groundset|)$ space for each coordinate and $O(|\supp(x)|) \cdot O(\log |\groundset|) = \tilde{O}(k/p)$ space in total. In addition to these coordinates, the vector $x$ might include a single non-zero coordinate taking the value of $k - p\lfloor k/p \rfloor$, storing the index and value of this coordinate require $O(\log |\groundset| + \log k) = \tilde{O}(1)$ space.
\end{proof}

We now divert our attention to analyzing the approximation ratio of Algorithm~\ref{alg:main_algorithm}. The first step in this analysis is lower bounding the value of $F(x)$, which we do by considering two cases, one when $\onenorm{\hat{x}}=k$, and the other when $\onenorm{\hat{x}}<k$. The following lemma bounds the value of $F(\hat{x})$ in the first of these cases.
\begin{lemma}
\label{lem:Fx-exactlyK}
	If $\onenorm{\hat{x}}=k$, then $F(\hat{x})\geq c\tau$.
\end{lemma}
\begin{proof}
	Denote by $u_1,u_2,\dots,u_\ell$ the elements that Algorithm~\ref{alg:main_algorithm} selects, in the order of their arrival. Using this notation, the value of $F(\hat{x})$ can be written as follows.
	\begin{align*}
		F(\hat{x})&=F(\onevect{\varnothing})+\sum_{i=1}^{\ell}{\Big(F\big(\hat{x}\wedge\onevect{\curly{u_1,u_2,\dotsc,u_i}}\big)-F\big(\hat{x}\wedge\onevect{\curly{u_1,u_2,\dotsc,u_{i-1}}}\big)\Big)}\\
		&=F(\onevect{\varnothing})+\sum_{i=1}^{\ell}{\Big(\hat{x}_{u_i}\cdot\partial_{u_i}F\big(\hat{x}\wedge\onevect{\curly{u_1,u_2,\dotsc,u_{i-1}}}\big)\Big)}\\
		&\geq{F(\onevect{\varnothing})}+\frac{c\tau}{k}\cdot\sum_{i=1}^{\ell}{\hat{x}_{u_i}}=F(\onevect{\varnothing})+\frac{c\tau}{k}\cdot\onenorm{\hat{x}}\geq c\tau\enspace,
	\end{align*}
where the second equality follows from the multilinearity of $F$, and the first inequality holds since Algorithm~\ref{alg:main_algorithm} selects an element $u_i$ only when $\partial_{u_i}F\big(\hat{x}\wedge\onevect{\curly{u_1,u_2,\dotsc,u_{i-1}}}\big)\geq\frac{c\tau}{k}$. The last inequality holds since $f$ (and thus, also $F$) is non-negative and $\|x\|_1 = k$ by the assumption of the lemma.
\end{proof}

We now consider the case in which $\onenorm{\hat{x}}<k$. Recall that our objective is to lower bound $F(\hat{x})$ in this case as well. Towards this goal, we bound the expression $F(\hat{x}+\onevect{OPT\setminus\supp(\hat{x})})$ from below and above in the next two lemmata.

\begin{lemma}
\label{lem:belowK_lower}
	If $\onenorm{\hat{x}}<k$, then $F\big(\hat{x}+\onevect{OPT\setminus\supp(\hat{x})}\big)\geq
	(1-p)\cdot\big[p\cdot f(OPT)+(1-p)\cdot f\big(OPT\setminus\supp(\hat{x})\big)\big]$.
\end{lemma}
\begin{proof}
	Since $\onenorm{\hat{x}}<k$, Observation~\ref{obs:x_type} guarantees that $x_u = p$ for every $u\in\supp(\hat{x})$, Thus $\hat{x}=p\cdot\onevect{OPT\cap\supp(\hat{x})}+p\cdot\onevect{\supp(\hat{x})\setminus OPT}$, and therefore,
	\begin{align*}
		F\big(\hat{x}+\onevect{OPT\setminus\supp(\hat{x})}\big)&=
		F\big(p\cdot\onevect{OPT\cap\supp(\hat{x})}
		+p\cdot\onevect{\supp(\hat{x})\setminus OPT}
		+\onevect{OPT\setminus\supp(\hat{x})}\big)\\
		&\geq (1-p)\cdot F\big(p\cdot\onevect{OPT\cap\supp(\hat{x})}+\onevect{OPT\setminus\supp(\hat{x})}\big)\\
		&\geq (1-p)\cdot \hat{f}\big(p\cdot\onevect{OPT\cap\supp(\hat{x})}+\onevect{OPT\setminus\supp(\hat{x})}\big)\\
		&=(1-p)\cdot\Big[p\cdot f(OPT)+(1-p)\cdot f\big(OPT\setminus\supp(\hat{x})\big)\Big]\enspace,
	\end{align*}
	where the first inequality follows from Corollary~\ref{cor:discardElements}, the second inequality hold
	since the Lov\'{a}sz extension lower bounds the multilinear extension (Lemma~\ref{lem:lovasz-multilinear}), and the last equality follows from the definition of the Lov\'{a}sz extension.
\end{proof}

\begin{lemma}
\label{lem:belowK_upper}
	If $\onenorm{\hat{x}}<k$, then $F\big(\hat{x}+\onevect{OPT\setminus\supp(\hat{x})}\big)\leq F(\hat{x})+c\tau$.
\end{lemma}
\begin{proof}
	The elements in $OPT\setminus\supp(\hat{x})$ were rejected by Algorithm~\ref{alg:main_algorithm}, which means that their marginal contribution with respect to the fractional solution $x$ at the time of their arrival was smaller than $c\tau/k$. Since the fractional solution $x$ only increases during the execution of the algorithm, the submodularity of $f$ guarantees that this is true also with respect to $\hat{x}$. More formally, we get
\[
	\partial_uF(\hat{x})<\frac{c\tau}{k}
	\quad \forall\; u\in  OPT \setminus \supp(\hat{x})
	\enspace.
\]
Using the submodularity of $f$ again, this implies
	\[
		F\big(\hat{x}+\onevect{OPT\setminus\supp(\hat{x})}\big)
			\leq F(\hat{x})+\mspace{-40mu}\sum_{u\in OPT\setminus\supp(\hat{x})}
			{\mspace{-40mu}\partial_uF(\hat{x})}
			\leq F(\hat{x})+|OPT \setminus \supp(\hat{x})| \cdot \frac{c\tau}{k}
			\leq F(\hat{x})+c\tau\enspace.		\qedhere
	\]
\end{proof}
Combining the last two lemmata immediately yields the promised lower bound on $F(\hat{x})$.
\begin{corollary}
\label{cor:belowK}
	If $\onenorm{\hat{x}}<k$, then $F(\hat{x})\geq
	(1-p)\cdot\Big[p\cdot f(OPT)+(1-p)\cdot f\big(OPT\setminus\supp(\hat{x})\big)\Big]-c\tau$.
\end{corollary}
Now that we have lower bounds on $F(\hat{x})$ for both cases, we can use them to get an expression for the approximation ratio of Algorithm~\ref{alg:main_algorithm}.

\begin{lemma}
\label{lem:bound-solution}
If $\tau \leq f(OPT)$, then
	$\Exp{\max\{f(S_1), f(S_2)\}}\geq\tau\cdot\min\curly{c,\frac{\alpha(1-p-c)}{\alpha+(1-p)^2}}$. In particular, for $c = \frac{\alpha(1-p)}{2\alpha+(1-p)^2}$ we get $\Exp{\max\{f(S_1), f(S_2)\}}\geq \frac{\alpha\tau(1-p)}{2\alpha+(1-p)^2}$
\end{lemma}
\begin{proof}
To see why the second part of the lemma follows from the first part, note that $c = \frac{\alpha(1-p)}{2\alpha+(1-p)^2}$ implies
\[
	\frac{\alpha(1-p-c)}{\alpha+(1-p)^2}
	=
	\frac{\alpha(1 - p)\left[1 - \frac{\alpha}{2\alpha + (1 - p)^2}\right]}{\alpha+(1-p)^2}
	=
	\frac{\alpha(1 - p)\left[\frac{2\alpha+(1-p)^2}{\alpha+(1-p)^2} - \frac{\alpha}{\alpha+(1-p)^2}\right]}{2\alpha+(1-p)^2}
	=
	\frac{\alpha(1 - p)}{2\alpha+(1-p)^2}
	\enspace.
\]

Given the above, we concentrate in the rest of the proof on proving the first part of the lemma. If $\|\hat{x}\| = k$, then by the definition of $S_1$ and Lemma~\ref{lem:Fx-exactlyK}
\[
	\Exp{\max\{f(S_1), f(S_2)\}}
	\geq
	\Exp{f(S_1)}
	\geq
	F(\hat{x})
	\geq
	c\tau
	\enspace.
\]
Consider now the case in which $\onenorm{\hat{x}}<k$. Note that $OPT \cap \supp(\hat{x})$ is a subset of the support of $\hat{x}$ of size at most $k$, and thus, by the definition of $S_2$, $\Exp{f(S_2)} \geq \alpha \cdot f(OPT \cap \supp(\hat{x}))$. Combining this with the lower bound given by Corollary~\ref{cor:belowK} for $\Exp{f(S_1)}$ when $\|x\|_1 < k$, we get
	\begin{align*}
		\Exp{\max\curly{f(S_1),f(S_2)}}\mspace{-81mu}&\mspace{81mu}
		\geq
		\max\curly{\Exp{f(S_1)},\Exp{f(S_2)}}\\
		\geq{}&
		\max\curly{(1-p)\cdot\big[p\cdot f(OPT)+(1-p)\cdot f(OPT\setminus\supp(\hat{x}))\big]-c\tau,\alpha\cdot{f(OPT\cap\supp(\hat{x}))}}\\
		\geq{} &
		\frac{\alpha}{\alpha+(1-p)^2}\cdot
		\bracks{(1-p)\cdot\bracks{p\cdot f(OPT)+(1-p)\cdot f(OPT\setminus\supp(\hat{x}))}-c\tau}\\
		&+\frac{(1-p)^2}{\alpha+(1-p)^2}\cdot\alpha\cdot f(OPT\cap\supp(\hat{x}))
		\enspace,
	\end{align*}
	where the last inequality holds since a maximum over two expressions is lower bounded by any convex combination of the two expressions in it. Rearranging the rightmost side of the last inequality yields
\[
	\frac{\alpha(1-p)^2}{\alpha+(1-p)^2}\cdot\Big[f(OPT\cap\supp(\hat{x}))+f(OPT\setminus\supp(\hat{x}))\Big]
		+\frac{\alpha\bracks{p(1-p)\cdot f(OPT)-c\tau}}{\alpha+(1-p)^2}\enspace.
\]
By assumption, $f(OPT) \geq \tau$. Additionally, by the submodularity and non-negativity of $f$, $f(OPT\cap\supp(\hat{x}))+f(OPT\setminus\supp(\hat{x})) \geq f(OPT) \geq \tau$. Combining all the above, we get
\[
	\Exp{\max\curly{f(S_1),f(S_2)}}
	\geq{} 
	\frac{\alpha(1-p)^2}{\alpha+(1-p)^2}\cdot\tau +\frac{\alpha\bracks{p(1-p)\cdot \tau-c\tau}}{\alpha+(1-p)^2}
	=
	\tau\cdot\frac{\alpha(1-p-c)}{\alpha+(1-p)^2}\enspace,
\]
	which completes the proof for the case of $\onenorm{\hat{x}}<k$.
\end{proof}

The following proposition summarizes the results we have proved so far.
\begin{proposition}
\label{prp:results}
As long as the algorithm used to compute $S_2$ runs in $\poly(\varepsilon) \cdot \tilde{O}(k)$ space, Algorithm~\ref{alg:main_algorithm} is a semi-streaming algorithm storing $O\left(k/p\right)$ elements. Moreover, for an appropriate choice of the parameter $c$, the output set produced by Algorithm~\ref{alg:main_algorithm} has an expected value of at least $\frac{\alpha\tau(1-p)}{2\alpha+(1-p)^2}$ whenever $\tau \leq f(OPT)$.
\end{proposition}

One consequence of Proposition~\ref{prp:results} is the following theorem. This theorem is similar to Theorem~\ref{trm:non-polynomial}, but it still assumes access to an estimate $\tau$ of $f(OPT)$. In Appendix~\ref{apx:estimatingTau} we show how to remove this assumption, using the technique of~\cite{DBLP:conf/kdd/BadanidiyuruMKK14}, which yields Theorem~\ref{trm:non-polynomial}.
\begin{theorem}
\label{trm:non-polynomial-knowing-tau}
	For every constant $\varepsilon \in (0, 1]$, there exists a semi-streaming algorithm that assumes access to an estimate $\tau$ of $f(OPT)$ obeying $(1-\varepsilon/8) \cdot f(OPT)\leq\tau\leq{f(OPT)}$ and provides $(3 +\varepsilon)$-approximation for the problem of maximizing a non-negative submodular function subject to cardinality constraint. This algorithm stores at most $O(k\varepsilon^{-1})$ elements.
\end{theorem}
\begin{proof}
In the theorem that we want to prove there is no restriction on the time complexity of the algorithm. Without such a restriction, Algorithm~\ref{alg:main_algorithm} can be implemented with $\alpha = 1$ (as discussed above). Consider the algorithm obtained from Algorithm~\ref{alg:main_algorithm} by setting $\alpha = 1$, $p = \varepsilon/8$ and $c$ as necessary to make Proposition~\ref{prp:results} hold. We show that this algorithm obeys all the requirements of the theorem. First, by Proposition~\ref{prp:results}, it stores at most $O(k/p) = O(k\varepsilon^{-1})$ elements, and it is a semi-streaming algorithm since the algorithm for calculating $S_2$ simply iterates over all subsets of the stored elements of size at most $k$, which requires no more than $O(k\varepsilon^{-1})$ space. Second, the expected value of the output set of this algorithm is at least
\[
	\frac{\alpha\tau(1-p)}{2\alpha+(1-p)^2}
	\geq
	\frac{(1 - \varepsilon/8) \cdot f(OPT) \cdot (1-\varepsilon/8)}{2+(1-\varepsilon/8)^2}
	\geq
	\frac{(1 - \varepsilon/4) \cdot f(OPT)}{3}
	\geq
	\frac{f(OPT)}{3 + \varepsilon}
	\enspace.
	\qedhere
\]
\end{proof}

Our next objective is to use Proposition~\ref{prp:results} to get also a guarantee for a polynomial time algorithm. To get a polynomial time implementation of Algorithm~\ref{alg:main_algorithm}, one has to handle two issues. The first issue is that one must use a polynomial time algorithm for calculating $S_2$, which leads to $\alpha < 1$. The second issue is related to the way the algorithm access the objective function. It is standard in the literature about submodular maximization to assume that algorithms have access to the objective function $f$ through a value oracle, which is an oracle that given a set $S \subseteq \groundset$ returns $f(S)$. For non-polynomial time algorithms, one can use this oracle to evaluate $F$ because the value of $F$ with respect to any given vector can be calculated using an exponential number of value oracle queries to $f$. However, for polynomial time algorithms there is no known way to do that. Thus, it is not clear how to implement Algorithm~\ref{alg:main_algorithm} using only value oracle access to $f$. We note, however, that it is easy to implement Algorithm~\ref{alg:main_algorithm}  using a value oracle access to $F$.\footnote{Observe that an algorithm that has value oracle access to $F$ can also evaluate derivatives of $F$ via the equality $\partial_uF(x) = F(x\vee\onevect{u})-F(x\wedge\onevect{\groundset\setminus\curly{u}})$, which holds for every vector $x \in [0, 1]^\groundset$ and element $u \in \groundset$.}

Keeping the two above issues in mind, we get the following theorem. This theorem is similar to Theorem~\ref{trm:polynomial}, but assumes access to an estimate $\tau$ of $f(OPT)$ and value oracle access to $F$. In Appendix~\ref{apx:estimatingTau} we explain how to remove the need to access $\tau$, and in  Appendix~\ref{apx:sampling} we explain how the value oracle access to $F$ can be replaced with a value oracle access to $f$ using standard sampling techniques, which completes the proof of Theorem~\ref{trm:polynomial}.
\begin{theorem}
\label{trm:polynomial-knowing-tau-oracle}
	There exists a semi-streaming $4.2819$-approximation algorithm for maximizing a non-negative submodular function $f$ subject to cardinality constraint that assumes
\begin{compactitem}
	\item value oracle access to the multilinear extension $F$ of $f$, and
	\item access to an estimate $\tau$ of $f(OPT)$ such that $\frac{9999}{10000} \cdot f(OPT)\leq\tau\leq{f(OPT)}$.
\end{compactitem}
The algorithm stores at most $O(k)$ elements.
\end{theorem}
\begin{proof}
We begin the proof by determining the value of $\alpha$ that we may assume in a polynomial time implementation of Algorithm~\ref{alg:main_algorithm}. In general, the state-of-the-art algorithm for the offline problem of maximizing a non-negative submodular function subject to a cardinality constraint is an algorithm of~\cite{to_appear:BuchbinderF19} achieving $2.598$-approximation for the problem. However, the instance of this problem that Algorithm~\ref{alg:main_algorithm} solves has additional structure. Specifically, the size of the ground set of this instance is $|\supp(x)| \leq \lceil k/p \rceil$. Furthermore, by guessing a single element of this ground set that does not belong to the optimal solution, we can assume that the algorithm needs to solve an offline instance in which the size of the ground set is upper bounded by $k/p$. For such instances, Buchbinder et al.~\cite{DBLP:conf/soda/BuchbinderFNS14} described an algorithm that for every constant $\varepsilon \in (0, 1]$ achieves in polynomial time an approximation ratio of
\[
	1 + \frac{\{\text{size of ground set of offline instance}\}}{2\sqrt{\left(\{\text{size of ground set of offline instance}\} - k\right)k}} + \varepsilon \enspace,
\]
when the size of the ground set of the offline instance is at least $2k$ and an approximation ratio of $2$ otherwise (this is not the way the approximation ratio of Buchbinder et al.~\cite{DBLP:conf/soda/BuchbinderFNS14} is stated in the original paper, but it follows from their proof). Plugging into this approximation ratio the upper bound we have on the size of the ground set of the offline instance, we get that for $p \leq 1/2$ the algorithm of~\cite{DBLP:conf/soda/BuchbinderFNS14} achieves an approximation ratio of at most
\[
	\max\left\{1 + \frac{k/p}{2\sqrt{(k/p - k)k}} + \varepsilon, 2\right\}
	=
	\max\left\{1 + \frac{1}{2\sqrt{p - p^2}} + \varepsilon, 2\right\}
	=
	1 + \frac{1}{2\sqrt{p - p^2}} + \varepsilon
	\enspace.
\]
Thus, for $p \leq 1/2$ we may assume in a polynomial time implementation of Algorithm~\ref{alg:main_algorithm} that $\alpha = [1 + 1/(2\sqrt{p - p^2})]^{-1} + \varepsilon$ for any constant $\varepsilon \in (0, 1]$.

Consider now the algorithm obtained from Algorithm~\ref{alg:main_algorithm} by setting $\alpha = 0.460675$, $p = 0.24$ and $c$ as necessary to make Proposition~\ref{prp:results} hold. One can verify that $\alpha > [1 + 1/(2\sqrt{p - p^2})]^{-1}$, and thus, the algorithm obtained in this way can be implemented in polynomial time according to the above discussion (assuming value oracle access to $F$). We also would like to show that the algorithm we obtained obeys the other requirements of the theorem we want to prove. First, by Proposition~\ref{prp:results}, it stores at most $O(k/p) = O(k)$ elements, and it is a semi-streaming algorithm since the algorithm of Buchbinder et al.~\cite{to_appear:BuchbinderF19} can be implemented to use only $\poly(\varepsilon^{-1}) \cdot \tilde{O}(k/p)$ space when run on a ground set of size $\lceil k/p \rceil$. Second, the expected value of the output set of this algorithm is at least
\begin{align*}
	\frac{\alpha\tau(1-p)}{2\alpha+(1-p)^2}
	\geq{} &
	\frac{0.460675 \cdot 0.9999 \cdot f(OPT) \cdot (1 - 0.24)}{2 \cdot 0.460675 + (1 - 0.24)^2}\\
	\geq{} &
	\frac{0.35007 \cdot f(OPT)}{1.49895}
	\geq
	0.233543 \cdot f(OPT)
	\enspace.
\end{align*}
The theorem now follows since $1 / 0.233543 \leq 4.2819$.
\end{proof}
\bibliography{SubmodularCardinality}
\bibliographystyle{plain}
\appendix
\section{Details about the Error in a Previous Work} \label{app:error}

As mentioned above, Chekuri et al.~\cite{DBLP:conf/icalp/ChekuriGQ15} described a semi-streaming algorithm for the problem of maximizing a non-negative (not necessarily monotone) submodular function subject to a cardinality constraint, and claimed an approximation ratio of $4.718$ for this algorithm. However, an error was later found in the proof of this result~\cite{communication:Chekuri18} (the error does not affect the other results of~\cite{DBLP:conf/icalp/ChekuriGQ15}). For completeness, we briefly describe in this appendix the error found.

In the proof presented by Chekuri et al.~\cite{DBLP:conf/icalp/ChekuriGQ15}, the output set of their algorithm is denoted by $\tilde{S}$. As is standard in the analysis of algorithms based on Sieve-Streaming, the analysis distinguishes between two cases: one case in which $\tilde{S} = k$, and a second case in which $\tilde{S} < k$. To argue about the second case, the analysis then implicitly uses the inequality
\begin{equation} \label{eq:needed_inequality}
	\Exp{f(\tilde{S} \cup OPT) \mid |S| < k}
	\geq
	(1 - \max_{u \in \groundset} \Pr[u \in \tilde{S}]) \cdot f(OPT)
	\enspace.
\end{equation}
It is claimed by~\cite{DBLP:conf/icalp/ChekuriGQ15} that this inequality follows from a lemma due to~\cite{DBLP:conf/soda/BuchbinderFNS14}. However, the lemma of~\cite{DBLP:conf/soda/BuchbinderFNS14} can yield only the inequalities
\[
	\Exp{f(\tilde{S} \cup OPT)}
	\geq
	(1 - \max_{u \in \groundset} \Pr[u \in \tilde{S}]) \cdot f(OPT)
\]
and
\[
	\Exp{f(\tilde{S} \cup OPT) \mid |S| < k}
	\geq
	(1 - \max_{u \in \groundset} \Pr[u \in \tilde{S} \mid |S| < k]) \cdot f(OPT)
	\enspace,
\]
which are similar to~\eqref{eq:needed_inequality}, but do not imply it.
\section{Inapproximability} \label{app:inapproximability}
In this appendix, we prove an inapproximability result for the problem of maximizing a non-negative submodular function subject to cardinality constraint in the data stream model. This result is given by the next theorem. The proof of the theorem is an adaptation of a proof given by Buchbinder et al.~\cite{to_appear:BuchbinderFS19} for a similar result applying to an online variant of the same problem.

\begin{theorem}
For every constant $\varepsilon > 0$, no data stream algorithm for maximizing a non-negative submodular function subject to cardinality constraint is $(2 - \varepsilon)$-competitive, unless it uses $\Omega(\power{\groundset})$ memory.
\end{theorem}
\begin{proof}
Let $k \geq 1$ and $h \geq 1$ be two integers to be chosen later, and consider the non-negative submodular function $\Func{f}{2^\groundset}{\realnum^+}$, where $\groundset=\curly{u_i}_{i=1}^{k-1}\cup\curly{v_i}_{i=1}^{h}\cup\curly{w}$,
	defined as follows.
	\[
		f(S)=
		\begin{cases}
			\power{S} &\text{if }w\notin S \enspace,\\
			k+\power{S\cap\curly{u_i}_{i=1}^{k-1}} &\text{if }w\in S\enspace.
		\end{cases}
	\]
It is clear that $f$ is non-negative. One can also verify that the marginal value of each element in $\groundset$ is non-increasing, and hence, $f$ is submodular. 

Let $ALG$ be an arbitrary data stream algorithm for the problem of maximizing a non-negative submodular function subject to a cardinality constraint, and let us consider what happens when we give this algorithm the above function $f$ as input, the last element of $\groundset$ to arrive is the element $w$ and we ask the algorithm to pick a set of size at most $k$. One can observe that, before the arrival of $w$, $ALG$ has no way to distinguish between the other elements of $\groundset$. Thus, if we denote by $M$ the set of elements stored by $ALG$ immediately before the arrival of $w$ and assume that the elements of $\groundset \setminus \{w\}$ arrive at a random order, then every element of $\groundset \setminus \{w\}$ belongs to $M$ with the same probability of $\Exp{|M|} / |\groundset \setminus \{w\}|$. Hence, there must exist some arrival order for the elements of $\groundset \setminus \{w\}$ guaranteeing that
\[
	\Exp{|M \cap \curly{u_i}_{i=1}^{k-1}|}
	=
	\sum_{i = 1}^{k - 1} \Pr[u_i \in M]
	\leq
	\frac{k \cdot \Exp{|M|}}{|\groundset \setminus \{w\}|}
	\enspace.
\]

Note now that the above implies that the expected value of the output set produced by $ALG$ given the above arrival order is at most
\[
	k + \frac{k \cdot \Exp{|M|}}{|\groundset \setminus \{w\}|}
	\enspace.
\]
In contrast, the optimal solution is the set $\curly{u_i}_{i=1}^{k-1}\cup\curly{w}$, whose value is $2k-1$. Therefore, the competitive ratio of $ALG$ is at least
\[
	\frac{2k - 1}{k + k \cdot \Exp{|M|} / |\groundset \setminus \{w\}|}
	=
	\frac{2 - 1/k}{1 + \Exp{|M|} / |\groundset \setminus \{w\}|}
	\geq
	2 - \frac{1}{k} - \frac{2 \cdot \Exp{|M|}}{|\groundset \setminus \{w\}|}
	\enspace.
\]
To prove the theorem we need to show that, when the memory used by $ALG$ is $o(|\groundset|)$, we can choose large enough values for $k$ and $h$ that will guarantee that the rightmost side of the last inequality is at least $2 - \varepsilon$. We do so by showing that the two terms $1/k$ and $2 \cdot \Exp{|M|} / |\groundset \setminus \{w\}|$ can both be upper bounded by $\varepsilon / 2$ when the integers $k$ and $h$ are large enough, respectively. For the term $1/k$ this is clearly the case when $k$ is larger than $2 / \varepsilon$. For the term $2 \cdot \Exp{|M|} / |\groundset \setminus \{w\}|$ this is true because increasing $h$ can make $\groundset$ as large as want, and thus, can make the ratio $\Exp{|M|} / |\groundset \setminus \{w\}|$ as small as necessary due to our assumption that the memory used by $ALG$ (which includes $M$) is $o(|\groundset|)$.
\end{proof}
\section{Estimating the Optimal Value}\label{apx:estimatingTau}
In this appendix, we explain how one can drop the assumption from Section~\ref{sec:algorithm} that the algorithm has access to an estimate $\tau$ of $f(OPT)$. This leads to versions of Theorems~\ref{trm:non-polynomial-knowing-tau} and~\ref{trm:polynomial-knowing-tau-oracle} without this assumption. Specifically, we prove the following two theorems. The first of these theorems is one of the results of this paper. In Appendix~\ref{apx:sampling} we explain how the proof of the second of these theorems can be modified to derive the other result of the paper (Theorem~\ref{trm:polynomial}).

\begin{reptheorem}{trm:non-polynomial}
	For every constant $\varepsilon \in (0, 1]$, there exists a semi-streaming $(3 +\varepsilon)$-approximation algorithm for maximizing a non-negative submodular function
	subject to cardinality constraint. The algorithm stores at most $O(k\varepsilon^{-2} \log k)$ elements.
\end{reptheorem}

\begin{theorem} \label{trm:polynomial-oracle-access}
	There exists a \emph{polynomial time} semi-streaming algorithm for maximizing a non-negative submodular function $f$
	subject to cardinality constraint that assumes value oracle access to the multilinear extension $F$ of $f$ and has an approximation ratio of at most $4.2819$. This algorithm stores at most $O(k \log k)$ elements.
\end{theorem}

The algorithm we use to prove Theorems~\ref{trm:non-polynomial} and~\ref{trm:polynomial-oracle-access} is Algorithm~\ref{alg:estimatingTau}. It gets the same two parameters $p$ and $c$ as Algorithm~\ref{alg:main_algorithm} plus an additional parameter $\varepsilon' \in (0, 1)$ controlling the quality guarantee of the output. The algorithm is based on a technique originally due to Badanidiyuru et al.~\cite{DBLP:conf/kdd/BadanidiyuruMKK14}. Throughout its execution, Algorithm~\ref{alg:estimatingTau} tracks in $m$ the maximum value of any singleton seen so far (or the value of the empty set if it is larger). The algorithm also maintains a set $T=\curly{(1+\varepsilon')^i\mid{m/ (1 + \varepsilon') \leq(1+\varepsilon')^i\leq{mk/c}}}$ of values that are either possible estimates for $OPT$ at the current point or might become such estimates in the future (of course, $T$ includes only a subset of the possible estimates). For every estimate $\tau$ in $T$, the algorithm maintains a fractional solution $x_\tau$. We note that the set of fractional solutions maintained is updated every time that $T$ is updated (which happens after every update of $m$). Specifically, whenever a new value $\tau$ is added to $T$, the algorithm instantiate a new vector $x_\tau$, and whenever a value $\tau$ is dropped from $T$, the algorithm deletes $x_\tau$.

While a value $\tau$ remains in $T$, Algorithm~\ref{alg:estimatingTau} maintains the fractional solution $x_\tau$ in exactly the same way that Algorithm~\ref{alg:main_algorithm} maintains its fractional solution given the value $\tau$ as an estimate for $f(OPT)$. Moreover, we show below that if $\tau$ remains in $T$ when the algorithm terminates, then the value of $x_\tau$ when the algorithm terminates is equal to the value of the vector $x$ when Algorithm~\ref{alg:main_algorithm} terminates after executing with $\tau$ as the estimate for $f(OPT)$. Thus, one can view Algorithm~\ref{alg:estimatingTau} as parallel execution of Algorithm~\ref{alg:main_algorithm} for many estimates of $f(OPT)$ at the same time. After viewing the last element, Algorithm~\ref{alg:estimatingTau} calculates for every $\tau \in T$ an output set $\hat{S}_\tau$ based on the fractional solution $\hat{x}_\tau$ in the same way Algorithm~\ref{alg:main_algorithm} does that, and then outputs the best output set computed for any $\tau \in T$.

\IncMargin{1em}
\begin{algorithm}[th]
\label{alg:estimatingTau}
\caption{Multilinear Threshold with No Access to $\tau$ $(p, c, \varepsilon')$}
	\DontPrintSemicolon
	Let $m\gets f(\varnothing)$ and $T\gets\curly{(1+\varepsilon')^h\mid{m / (1 + \varepsilon') \leq(1+\varepsilon')^h\leq{mk/c}}}$.\\
	\For {each arriving element $u$} {
		\If {$m<f(\curly{u})$} {
			Update $m\gets{f(\curly{u})}$ and $T\gets\curly{(1+\varepsilon')^h\mid{m/ (1 + \varepsilon') \leq(1+\varepsilon')^h\leq{mk/c}}}$. \label{line:T_update_tau}\\
			Delete $x_{\tau}$ for every value $\tau$ removed from $T$ in Line~\ref{line:T_update_tau}.\\
			Initialize $x_{\tau}\gets\onevect{\varnothing}$ for every value $\tau$ added to $T$ in Line~\ref{line:T_update_tau}.
		}
		\For {every $\tau\in{T}$} {
			\lIf {$\partial_uF(x_{\tau})\geq\frac{c\tau}{k}$} {
				$x_{\tau}\gets{x_{\tau}}+\min\curly{p,k-\onenorm{x_{\tau}}}\cdot\onevect{u}$.
			}
		}
	}
	\For {every $\tau\in{T}$} {
		Round the vector $x_{\tau}$ to yield a feasible solution $S^\tau_1$ such that $\Exp{f(S^\tau_1)}\geq{F(x_{\tau})}$.\\
		Find a feasible solution $S^\tau_2\subseteq\supp(x_{\tau})$ such that $\Exp{f(S^\tau_2)} \geq \alpha \cdot \arg \max_{S \subseteq \supp(x_{\tau}), |S| \leq k} f(S)$.\\
		Let $\hat{S}_\tau$ be the better solution among $S^\tau_1$ and $S^\tau_2$.
	}
	\Return{the best solution among $\{\hat{S}_\tau\}_{\tau \in T}$, or the empty set if $T = \varnothing$}.
\end{algorithm}\DecMargin{1em}

We begin the analysis of Algorithm~\ref{alg:estimatingTau} by bounding its space complexity.
\begin{observation} \label{obs:space_tau}
Assuming it takes $O(1)$ space to store an element of $\groundset$ and a value returned by $F$, Algorithm~\ref{alg:estimatingTau} can be implemented so that it stores at most $O(kp^{-1}(\varepsilon')^{-1}(\ln k - \ln c))$ elements and its space complexity is $\tilde{O}(kp^{-1}(\varepsilon')^{-1}(- \ln c))$, excluding the space complexity of the algorithm used to find the sets $S^\tau_2$.
\end{observation}
\begin{proof}
The number of estimates in $T$ is upper bounded at all times by
\[
	1 + \log_{1 + \varepsilon'} \left(\frac{km/c}{m/(1 + \varepsilon')}\right)
	=
	1 + \frac{1 + \ln k - \ln c}{\ln(1 + \varepsilon')}
	\leq
	1 + \frac{1 + \ln k - \ln c}{2\varepsilon'/3}
	=
	O((\varepsilon')^{-1}(\ln k - \ln c))
	\enspace.
\]
Algorithm~\ref{alg:estimatingTau} maintains for every $\tau \in T$ the same information maintained by Algorithm~\ref{alg:main_algorithm}, which requires $O(k/p)$ elements and $\tilde{O}(k/p)$ space for every $\tau \in T$, or equivalently, $O(kp^{-1}(\varepsilon')^{-1}(\ln k - \ln c))$ elements and $\tilde{O}(kp^{-1}(\varepsilon')^{-1}(- \ln c))$ space for all the values in $T$ together. In addition to this information, the algorithm also has to store $m$, which requires constant space. We note that there is no need to explicitly store $T$ because the estimates added to it or removed from it in every update of $m$ can be easily determined using the old and new values of $m$.
\end{proof}

Our next objective is to show that the approximation guarantee of Algorithm~\ref{alg:main_algorithm} extends to Algorithm~\ref{alg:estimatingTau}. Let $\hat{m}$ and $\hat{T}$ be the final values of $m$ and $T$, respectively. We assume in this section that $c$ is set to the value given in Lemma~\ref{lem:bound-solution}---the value $\frac{\alpha(1 - p)}{2\alpha + (1 - p)^2}$. Recall that this is the value required to make Proposition~\ref{prp:results} hold. Using this choice of $c$ allows us to easily handle the rare case in which $\hat{T}$ is empty.

\begin{observation}
\label{obs:T-not-empty}
$c \in (0, 1/2]$, and thus, $\hat{T}$ is not empty unless $\hat{m} = 0$.
\end{observation}
\begin{proof}
Since $\alpha \in (0, 1]$ and $p \in (0, 1)$,
\[
	c
	=
	\frac{\alpha(1 - p)}{2\alpha + (1 - p)^2}
	\in
	\left(0, \frac{\alpha}{2\alpha}\right]
	=
	(0, 1/2]
	\enspace.
	\qedhere
\]
\end{proof}

The last observation immediately implies that when $\hat{T}$ is empty, all the singletons have zero values. Thus, both $OPT$ and the empty set have zero values, which makes the output of Algorithm~\ref{alg:estimatingTau} optimal in this case. Hence, we can assume from now on that $\hat{T} \neq \varnothing$. The following lemma shows that $\hat{T}$ contains a good estimate for $f(OPT)$ in this case.
\begin{lemma}
\label{lem:apxtau-tau-in-T}
	The set $\hat{T}$ contains a value $\hat{\tau}$ such that $(1-\varepsilon') \cdot f(OPT)\leq\hat{\tau}\leq f(OPT)$.
\end{lemma}
\begin{proof}
Observe that $\hat{m}=\max\curly{f(\varnothing), \max_{u\in\groundset}\curly{f\big(\curly{u}\big)}}$. Thus, by the submodularity of $f$,
\[
	f(OPT)
	\leq
	f(\varnothing)
	+
	\sum_{u \in OPT} \mspace{-9mu} \bracks{f(\{u\}) - f(\varnothing)}
	\leq
	\max\curly{f(\varnothing), \sum_{u \in OPT} \mspace{-9mu} f\big(\{u\}\big)}
	\leq
	k\hat{m}
	\leq
	\frac{k\hat{m}}{c}
	\enspace.
\]
In contrast, by the definition of $OPT$,
\[
	f(OPT)
	\geq
	\max\curly{f(\varnothing), \max_{u\in\groundset}\curly{f\big(\curly{u}\big)}}
	=
	\hat{m}
	\enspace.
\]
Since $\hat{T}$ contains all the values of the form $(1 + \varepsilon')^i$ in the range $[\hat{m} / (1 + \varepsilon'), k\hat{m}/c]$, the above inequalities imply that it contains in particular the largest value of this form that is still not larger than $f(OPT)$. Let us denote this value by $\hat{\tau}$. By definition, $\hat{\tau} \leq f(OPT)$. Additionally,
\[
	\hat{\tau} \cdot (1 + \varepsilon') \geq f(OPT)
	\Rightarrow
	\hat{\tau} \geq \frac{f(OPT)}{1 + \varepsilon'} \geq (1 - \varepsilon') \cdot f(OPT)
	\enspace.
	\qedhere
\]
\end{proof}

Let us now concentrate on the value $\hat{\tau}$ whose existence is guaranteed by Lemma~\ref{lem:apxtau-tau-in-T}, and let $\bar{x}$ denote the fractional solution maintained by Algorithm~\ref{alg:main_algorithm} when it gets $\hat{\tau}$ as the estimate for $f(OPT)$. Additionally, let us denote by $u_1,u_2,\dots,u_n$ the elements of $\groundset$ in the order of their arrival, and let $u_j$ be the element whose arrival caused the addition of $\hat{\tau}$ to $T$, i.e., $u_j$ is the first element satisfying $\hat{\tau}\leq (k/c) \cdot f(\curly{u_j})$ (if $\hat{\tau} \in T$ from the very beginning, then we define $j = 0$). The following lemma shows that, prior to the arrival of $u_j$, the fractional solution $\bar{x}$ of Algorithm~\ref{alg:main_algorithm} was empty.
\begin{lemma}
\label{lem:apxtau-before-tau}
	For every integer $1 \leq t \leq j - 1$, $\partial_{u_t}F\big(\bar{x}\wedge\onevect{\curly{u_1,u_2,\dotsc,u_{t-1}}}\big)<\frac{c\hat{\tau}}{k}$, and thus, no fraction of $u_t$ was added to $\bar{x}$.
\end{lemma}
\begin{proof}
	If $j = 0$, then the lemma is trivial. Otherwise, by the definition of $u_j$ as the first element obeying $\hat{\tau}\leq (k/c) \cdot f(\curly{u_j})$,
	\[
		\frac{c\hat{\tau}}{k}
		>
		f(\{u_t\})
		\geq
		f(\{u_t\}) - f(\varnothing)
		=
		\partial_{u_t} F(\onevect{\varnothing})
		\geq
		\partial_{u_t}F\big(\bar{x}\wedge\onevect{\curly{u_1,u_2,\dotsc,u_{t-1}}}\big)
		\enspace,
	\]
	where the second inequality follows from the non-negativity of $f$ and the last from its submodularity.
\end{proof}

According to the above discussion, from the moment $\hat{\tau}$ gets into $T$, Algorithm~\ref{alg:estimatingTau} updates the fractional solution $x_{\hat{\tau}}$ in the same way that Algorithm~\ref{alg:main_algorithm} updates $\bar{x}$ (note that, once $\hat{\tau}$ gets into $T$, it remains there for good since $\hat{\tau} \in \hat{T}$). Together with the previous lemma which shows that $\bar{x}$ is empty just like $x_{\hat{\tau}}$ at the moment $\hat{\tau}$ gets into $T$---which is also the moment of the arrival of $u_j$, this implies that the final value of $x_{\hat{\tau}}$ is equal to the final value of $\bar{x}$. Since the set $\hat{S}_{\hat{\tau}}$ is computed based on the final value of $x_{\hat{\tau}}$ in the same way that the output of Algorithm~\ref{alg:main_algorithm} is computed based on the final value of $\bar{x}$, we get the following corollary.

\begin{corollary} \label{cor:reduction1-2}
If it is guaranteed that the approximation ratio of Algorithm~\ref{alg:main_algorithm} is at least $\beta$ when $(1 - \varepsilon') \cdot f(OPT) \leq \tau \leq f(OPT)$ for some choice of the parameters $p$ and $c$, then the approximation ratio of Algorithm~\ref{alg:estimatingTau} is at most $\beta$ as well for this choice of $p$ and $c$.
\end{corollary}

We are now ready to prove Theorems~\ref{trm:non-polynomial} and~\ref{trm:polynomial-oracle-access}.

\begin{proof}[Proof of Theorem~\ref{trm:non-polynomial}]
The proof of Theorem~\ref{trm:non-polynomial-knowing-tau} shows that Algorithm~\ref{alg:main_algorithm} achieves an approximation ratio of $3 + \varepsilon$ when it has access to a value $\tau$ obeying $(1-\varepsilon/8) \cdot f(OPT)\leq\tau\leq{f(OPT)}$ and its parameters are set to $\alpha = 1$, $p = \varepsilon / 8$ and $c = \frac{\alpha(1 - p)}{2\alpha + (1 - p)^2}$. According to Corollary~\ref{cor:reduction1-2}, this implies that setting the parameters $\alpha$, $p$ and $c$ of Algorithm~\ref{alg:estimatingTau} in the same way and setting $\varepsilon'$ to $\varepsilon / 8$, we get an algorithm whose approximation ratio is at most $3 + \varepsilon$ and does not assume access to an estimate of $f(OPT)$.

It remains to bound the space requirements of the algorithm obtained in this way. Observe that
\[
	c
	=
	\frac{\alpha(1 - p)}{2\alpha + (1 - p)^2}
	=
	\frac{1 - \varepsilon/8}{2 + (1 - \varepsilon/8)^2}
	\geq
	\frac{1/2}{3}
	=
	\frac{1}{6}
	\enspace.
\]
Plugging this bound and the equality $p = \varepsilon / 8$ into the guarantee of Observation~\ref{obs:space_tau}, we get that the algorithm we obtained stores at most $O(k\varepsilon^{-2} \log k)$ elements, and uses $\tilde{O}(k\varepsilon^{-2})$ space (since the algorithm for calculating $S_2$ uses only $O(k \varepsilon^{-1})$ space as explained in the proof of Theorem~\ref{trm:non-polynomial-knowing-tau}), which implies that it is a semi-streaming algorithm.
\end{proof}

\begin{proof}[Proof of Theorem~\ref{trm:polynomial-oracle-access}]
The proof of Theorem~\ref{trm:polynomial-knowing-tau-oracle} shows that Algorithm~\ref{alg:main_algorithm} runs in polynomial time and achieves an approximation ratio of $4.2819$ when it has access to a value $\tau$ obeying $(1-10^{-4}) \cdot f(OPT)\leq\tau\leq{f(OPT)}$, it has value oracle access to $F$ and its parameters are set to $\alpha = 0.460675$, $p = 0.24$ and $c = \frac{\alpha(1 - p)}{2\alpha + (1 - p)^2}$. Since Algorithm~\ref{alg:estimatingTau} requires only a polynomial amount of time on top of the time required by $|T|$ instances of Algorithm~\ref{alg:main_algorithm}, this implies that Algorithm~\ref{alg:estimatingTau} can also be implemented to run in polynomial time (given value oracle access to $F$) when the parameters $p$, $c$ and $\alpha$ are set as above and $\varepsilon'$ is set to $10^{-4}$. Moreover, Corollary~\ref{cor:reduction1-2} guarantees that, after setting the parameters in this way, the approximation ratio of Algorithm~\ref{alg:estimatingTau} is at most $4.2819$.

It remains to bound the space required by Algorithm~\ref{alg:estimatingTau} when the parameters are set as above. Observe that
\[
	c
	=
	\frac{\alpha(1 - p)}{2\alpha + (1 - p)^2}
	=
	\frac{0.460675 \cdot (1 - 0.24)}{2 \cdot 0.460675 + (1 - 0.24)^2}
	=
	\frac{0.3501092}{1.49895}
	\geq
	0.233
	\enspace.
\]
Plugging this bound and the equalities $p = 0.24$ and $\varepsilon' = 10^{-4}$ into the guarantee of Observation~\ref{obs:space_tau}, we get that Algorithm~\ref{alg:estimatingTau} with the above parameter values stores at most $O(k \log k)$ elements, and uses $\tilde{O}(k)$ space in addition to the $\poly(1/\varepsilon') \cdot \tilde{O}(k)$ space used by the algorithm for calculating $S^\tau_2$ (as explained in the proof of Theorem~\ref{trm:polynomial-knowing-tau-oracle}), which implies that it is a semi-streaming algorithm.
\end{proof}
\section{Approximating the Multilinear Extension} \label{apx:sampling}

In this appendix we prove Theorem~\ref{trm:polynomial} by presenting a polynomial time variant of Algorithm~\ref{alg:estimatingTau}, given below as Algorithm~\ref{alg:sampling}. We observe that the only difference between Algorithms~\ref{alg:estimatingTau} and~\ref{alg:sampling} is that the latter algorithm uses estimates for the partial derivatives of $F$ instead of the actual values of these derivatives (which are difficult to compute in polynomial time). These estimates are calculated in Line~\ref{line:estimates_calculation} of the algorithm.

\IncMargin{1em}
\begin{algorithm}[ht]
\label{alg:sampling}
\caption{\texttt{Multilinear Threshold with No Oracle Access to $F$} $(p,c,\varepsilon')$}
	
	\DontPrintSemicolon
	Let $m\gets f(\varnothing)$ and $T\gets\curly{(1+\varepsilon')^h\mid{m/(1 + \varepsilon')\leq(1+\varepsilon')^h\leq{mk/c}}}$.\\
	\For {each arriving element $u_i$} {
		\If {$m<f(\curly{u_i})$} {
			Update $m\gets{f(\curly{u_i})}$ and $T\gets\curly{(1+\varepsilon')^h\mid{m/(1 + \varepsilon')\leq(1+\varepsilon')^h\leq{mk/c}}}$. \label{line:T_update_F}\\
			Delete $x_\tau$ for every value $\tau$ removed from $T$ in Line~\ref{line:T_update_F}.\\
			Initialize $x_{\tau}\gets\onevect{\varnothing}$ for every value $\tau$ added to $T$ in Line~\ref{line:T_update_F}.
		}
		\For {every $\tau\in{T}$} {
			Let $\bar\partial^\tau_{u_i}F(x_\tau)$ be an approximation of
			\[
			\partial_{u_i}F(x_\tau)
			={} F(x_\tau\vee\onevect{u_i})
			-F(x_\tau\wedge\onevect{\groundset\setminus{u_i}})
			={} \Exp{f\big(\RFunc\left(x_\tau\right)\cup\curly{u_i}\big)
			-f\big(\RFunc(x_\tau)\big)}
			\]
			obtained by averaging $\ell=\ceil{\frac{4800(p^{-1} + 1)^2k^2}{[\varepsilon'(1-\varepsilon')]^2} \cdot \ln\left(80i^2(\varepsilon')^{-1}\right)}$ samples. \label{line:estimates_calculation}\\
			\If {$\bar\partial^\tau_{u_i}F(x_\tau)\geq\frac{c\tau}{k}$} {
				$x_\tau\gets{x_\tau}+\min\curly{p,k-\onenorm{x_\tau}}\cdot\onevect{u_i}$.
			}
		}
	}
	\For {every $\tau\in{T}$} {
		Round the vector $x_\tau$ to yield a feasible solution $S^\tau_1$ such that $\Exp{f(S^\tau_1)}\geq{F(x_\tau)}$.\\
		Find a feasible solution $S^\tau_2\subseteq\supp(x_\tau)$ such that $\Exp{f(S^\tau_2)} \geq \alpha \cdot \arg \max_{S \subseteq \supp(x_\tau), |S| \leq k} f(S)$.\\
		Let $\hat{S}_\tau$ be the better solution among $S^\tau_1$ and $S^\tau_2$.
	}
	\Return{the best solution among $\{\hat{S}_\tau\}_{\tau \in T}$, or the empty set if $T = \varnothing$}.
\end{algorithm}\DecMargin{1em}

We begin the analysis of Algorithm~\ref{alg:sampling} by bounding its space complexity.

\begin{lemma}
\label{lem:apxF-space}
	Assuming it takes $O(1)$ space to store an element of $\groundset$ and a value returned by $f$, Algorithm~\ref{alg:sampling} can be implemented so that it stores at most $O(kp^{-1}(\varepsilon')^{-1}(\ln k - \ln c))$ elements and its space complexity is $\tilde{O}(kp^{-1}(\varepsilon')^{-1}(- \ln c) + \log \varepsilon^{-1} + \log (1 - \varepsilon)^{-1} + \log p^{-1})$, excluding the space complexity of the algorithm used to find the sets $S^\tau_2$.
\end{lemma}
\begin{proof}
Observe that apart from the space used to calculate the estimates of the derivatives, Algorithms~\ref{alg:estimatingTau} and~\ref{alg:sampling} share the same space complexity. Thus, in this proof we only bound the space required for computing the estimates.
	
To calculate each estimate, Algorithm~\ref{alg:sampling} has to store the sum of $\ell$ samples. Since we assume that each sample can be stored in constant space, storing this sum requires
	\begin{align*}
		O(\log \ell)
		={} &
		O\left(\log \left(\ceil{4800\left[\varepsilon'(1-\varepsilon')pk^{-1}\right]^{-2}\ln\left(80i^2(\varepsilon')^{-1}\right)}\right)\right)\\
		={} &
		O\left(\log k + \log \varepsilon^{-1} + \log (1 - \varepsilon)^{-1} + \log p^{-1} + \log \log (i^2/\varepsilon')\right)\\
		={} &
		\tilde{O}(\log \varepsilon^{-1} + \log (1 - \varepsilon)^{-1} + \log p^{-1})
		\enspace,
	\end{align*}
	where the last equality holds since $i$ is upper bounded by $n$ and the $\tilde{O}$ notation suppresses terms that are poly-logarithmic in $n$ and $k$. Since Algorithm~\ref{alg:sampling} need to store only one estimate at each time point, its space complexity exceeds the space complexity of Algorithm~\ref{alg:estimatingTau} only by the above expression.
\end{proof}

Our next objective is to analyze the approximation guarantee of Algorithm~\ref{alg:sampling}.
We note that the proofs of Observation~\ref{obs:T-not-empty} and Lemmata~\ref{lem:apxtau-tau-in-T} and~\ref{lem:apxtau-before-tau} apply also to Algorithm~\ref{alg:sampling} without any change. Thus, we know that Algorithm~\ref{alg:sampling} outputs an optimal solution if the final set $T$ is empty (so we assume from now that it is not), and that there exists a value $\hat{\tau}$ and integer $1 \leq j \leq n$ such that
\begin{compactitem}
	\item $\hat{\tau}$ enters $T$ when $u_j$ arrives (unless $\hat{\tau}$ belongs to $T$ from the very beginning of the algorithm, in which case we define $j = 1$),
	\item once $\hat{\tau}$ enters into $T$, it remains there until the algorithm terminates,
	\item for every $1 \leq t < j$, $\partial_{u_t} F(\tx \wedge \onevect{u_1, u_2, \dotsc, u_{t - 1}}) < c\hat{\tau}/k$, where $\tx$ denotes the final value of fractional solution $x_{\hat\tau}$,
	\item and $(1-\varepsilon') \cdot f(OPT)\leq\hat{\tau}\leq f(OPT)$.
\end{compactitem}
It is important to observe also that the value of $x_{\hat\tau}$ when the element $u_i$ arrives is $\tx\wedge\onevect{\curly{u_1,u_2,\dots,u_{i-1}}}$. Let us denote now by $\cE$ the event that the estimate $\bar\partial^{\hat{\tau}}_{u_i}F\left(\tx\wedge\onevect{\curly{u_1,u_2,\dots,u_{i-1}}}\right)$ that is calculated by Algorithm~\ref{alg:sampling} differs from $\partial_{u_i}F\left(\tx\wedge\onevect{\curly{u_1,u_2,\dots,u_{i-1}}}\right)$ by at most $\frac{\varepsilon'(1-\varepsilon')}{20k}\cdot f(OPT)$ for every $j\leq i\leq n$ (i.e., $\power{\bar\partial^{\hat{\tau}}_{u_i}F\left(\tx\wedge\onevect{\curly{u_1,u_2,\dots,u_{i-1}}}\right)-\partial_{u_i}F\left(\tx\wedge\onevect{\curly{u_1,u_2,\dots,u_{i-1}}}\right)}\leq\frac{\varepsilon'(1-\varepsilon')}{20k}\cdot f(OPT)$ for every $1 \leq i \leq n$). Intuitively, $\cE$ is the event that all the estimates done by Algorithm~\ref{alg:sampling} with respect to $\hat{\tau}$ are quite accurate. In the next few claims we show that $\cE$ is a high probability event. We first need the following known Chernoff-like lemma.

\begin{lemma}[Lemma B.3 from~\cite{to_appear:BuchbinderFS19}]
\label{lem:apxF-chernoff}
	Let $X_1,X_2,\dots,X_\ell$ be independent random variables such that for each $i$, $X_i\in\range{-1}{1}$. Let $X=\frac{1}{\ell}\sum_{i=1}^\ell{X_i}$ and $\mu=\Exp{X}$. Then
	\[ \pr{X>\mu+\alpha}\leq e^{-\frac{\alpha^2\ell}{12}} \mspace{50mu}\text{and}\mspace{50mu} \pr{X<\mu-\alpha}\leq e^{-\frac{\alpha^2\ell}{8}} \]
	for every $\alpha>0$.
\end{lemma}

In the next lemma we show that each estimate that Algorithm~\ref{alg:sampling} calculates is not likely to be too far away from the actual value of the derivative.
\begin{lemma}
\label{lem:apxF-prob-each-i}
	For every $j\leq i\leq n$, $\Pr\Big[\power{\bar\partial^{\hat{\tau}}_{u_i}F\left(\tx\wedge\onevect{\curly{u_1,u_2,\dots,u_{i-1}}}\right)-\partial_{u_i}F\left(\tx\wedge\onevect{\curly{u_1,u_2,\dots,u_{i-1}}}\right)}>\frac{\varepsilon'(1-\varepsilon')}{20k}\cdot f(OPT)\Big]\leq\frac{\varepsilon'}{40i^2}$.
\end{lemma}
\begin{proof}
		The algorithm calculates the estimate $\bar\partial^{\hat{\tau}}_{u_i}F\left(\tx\wedge\onevect{\curly{u_1,u_2,\dots,u_{i-1}}}\right)$ by averaging $\ell$ samples of
	$f\big(\RFunc(\tx\wedge\onevect{\curly{u_1,u_2,\dots,u_{i-1}}})\cup\curly{u_i}\big)-f\big(\RFunc(\tx\wedge\onevect{\curly{u_1,u_2,\dots,u_{i-1}}})\big)$. Let $Y_t$ be the $t$-th such sample. Note that the expected value of $Y_t$ for every $1 \leq t \leq \ell$ is exactly $\partial_{u_i}F\left(\tx\wedge\onevect{\curly{u_1,u_2,\dotsc,u_{i-1}}}\right)$. Later in this proof, we will show that the value of $Y_t$ always fall within the range $\range{-(p^{-1} + 1)\cdot f(OPT)}{f(OPT)}$. However, before doing so, let us prove that the lemma holds under this assumption.

Consider the random variables $Y_1 / [(p^{-1} + 1) \cdot f(OPT)], Y_2 / [(p^{-1} + 1) \cdot f(OPT)], \dotsc, Y_\ell / [(p^{-1} + 1) \cdot f(OPT)]$. Due to the above assumption, the values of these random variables always fall within the range $\range{-1}{1}$, therefore, by Lemma~\ref{lem:apxF-chernoff},
	\begin{align*}
		\mspace{100mu}&\mspace{-100mu}\pr{\power{\bar\partial^{\hat{\tau}}_{u_i}F\left(\tx\wedge\onevect{\curly{u_1,u_2,\dots,u_{i-1}}}\right)
		-\partial_{u_i}F\left(\tx\wedge\onevect{\curly{u_1,u_2,\dots,u_{i-1}}}\right)}
		>{} \frac{\varepsilon'(1-\varepsilon')}{20k}\cdot f(OPT)}\\
		={} &
		\pr{\power{\frac{1}{\ell} \cdot \sum_{t = 1}^\ell Y_t - \Exp{\frac{1}{\ell} \cdot \sum_{t = 1}^\ell Y_t}} > \frac{\varepsilon'(1-\varepsilon')}{20k}\cdot f(OPT)}\\
		={} &
		\pr{\power{\frac{1}{\ell} \cdot \sum_{t = 1}^\ell \frac{Y_t}{(p^{-1} + 1) \cdot f(OPT)} - \Exp{\frac{1}{\ell} \cdot \sum_{t = 1}^\ell \frac{Y_t}{(p^{-1} + 1) f(OPT)}}} > \frac{\varepsilon'(1-\varepsilon')}{20k(p^{-1} + 1)}}\\
		\leq{} &
		2e^{-\frac{\ell\bracks{\frac{\varepsilon'(1-\varepsilon')}{20k(p^{-1} + 1)}}^2}{12}}
		\leq{}
		2e^{\ln\frac{\varepsilon'}{80i^2}}
		={}
		\frac{\varepsilon'}{40i^2}\enspace.
	\end{align*}
	
	It only remains to prove that for every $1\leq t\leq \ell$, the sample $Y_t$ is contained within the range $\range{-(p^{-1} + 1)\cdot f(OPT)}{f(OPT)}$. Since $f$ is submodular and $f(OPT)$ upper bounds the value of every set of $k$ elements with respect to $f$, every set $S$ of at most $\lceil p^{-1}k \rceil \leq (p^{-1} + 1)k$ elements obeys $f(S)\leq (p^{-1} + 1)\cdot f(OPT)$. Therefore,
	\[ Y_t
	\geq{} -f\left(\RFunc\left(\tx\wedge\onevect{\curly{u_1,u_2,\dots,u_{i-1}}}\right)\right)
	\geq{} -(p^{-1} + 1)\cdot f(OPT) \]
	because $\tx$ contains at most $\lceil p^{-1}k\rceil$ non-zero entries.
	In contrast, the submodularity of $f$ also implies
	\begin{align*}
		Y_t
		=
		f\left(\RFunc\left(\tx\wedge\onevect{\curly{u_1,u_2,\dots,u_{i-1}}}\right)\cup\curly{u_i}\right)
		-f\left(\RFunc\left(\tx\wedge\onevect{\curly{u_1,u_2,\dots,u_{i-1}}}\right)\right)
		\leq{} &
		f(\curly{u_i})-f(\varnothing)\\
		\leq{} &
		f(\curly{u_i})
		\leq{} f(OPT)\enspace.\enspace\qedhere
	\end{align*}
\end{proof}

Lemma~\ref{lem:apxF-prob-each-i} implies the following corollary.

\begin{corollary}
\label{cor:apxF-prob-all-i}
	$\pr{\cE}\geq 1 - \varepsilon'/20$.
\end{corollary}
\begin{proof}
	By Lemma~\ref{lem:apxF-prob-each-i} and the union bound,
	\[ 1 - \pr{\cE}\leq
		\sum_{i=j}^n\frac{\varepsilon'}{40i^2}
		={} \frac{\varepsilon'}{40}\sum_{i=j}^n\frac{1}{i^2}
		\leq{} \frac{\varepsilon'}{40}+\frac{\varepsilon'}{40}\int_1^n\frac{1}{x^2}dx
		={} \frac{\varepsilon'}{40}-\frac{\varepsilon'}{40}\bracks{\frac{1}{x}}_1^n
		={} \frac{\varepsilon'}{40}-\frac{\varepsilon'}{40n}+\frac{\varepsilon'}{40}
		\leq{} \frac{\varepsilon'}{20}\enspace. \qedhere \]
	
\end{proof}

In the next few lemmata we show that, given that the event $\cE$ happens, the approximation ratio of Algorithm~\ref{alg:sampling} is good.
We first consider the case in which $\onenorm{\tx}=k$. The following lemma corresponds to Lemma~\ref{lem:Fx-exactlyK} and bounds the value of $F(\tx)$ from below in this case.

\begin{lemma}
\label{lem:apxF-exactlyK}
Assuming $\cE$ holds, if $\onenorm{\tx}={} k$, then $F(\tx)\geq{} \hat\tau\cdot(c-\varepsilon'/20)$.
\end{lemma}
\begin{proof}
	Note that since $\tx$ contains a zero fraction of every element that arrived before $u_j$, the value of $F(\tx)$ can also be written as
	\begin{align*}
		F(\tx)
		={} &
		F(\onevect{\varnothing})+\sum_{i=j}^n\Big(F\big(\tx\wedge\onevect{\curly{u_1,u_2,\dotsc,u_i}}\big)
		-F\big(\tx\wedge\onevect{\curly{u_1,u_2,\dotsc,u_{i-1}}}\big)\Big)\\
		={} &
		F(\onevect{\varnothing})+\sum_{i=j}^n\Big(\tx_{u_i}\cdot\partial_{u_i}F\big(\tx\wedge\onevect{\curly{u_1,u_2,\dotsc,u_{i-1}}}\big)\Big)\enspace,
	\end{align*}
	where the second equality follows from the multilinearity of $F$. Since we assume that the event $\cE$ happened, for every $j \leq i \leq n$ it holds that
	\begin{align*} \bar\partial^{\hat{\tau}}_{u_i}(\tx\wedge\onevect{\curly{u_1,u_2,\dotsc,u_{i-1}}})
	\leq{}& \partial_{u_i}(\tx\wedge\onevect{\curly{u_1,u_2,\dotsc,u_{i-1}}})+\frac{\varepsilon'(1-\varepsilon')}{20k}\cdot f(OPT)\\
	\leq{}& \partial_{u_i}(\tx\wedge\onevect{\curly{u_1,u_2,\dotsc,u_{i-1}}})+\hat\tau\cdot\frac{\varepsilon'}{20k}\enspace, \end{align*}
	where the second inequality follows since $(1-\varepsilon')\cdot f(OPT)\leq\hat\tau$. Combining this with the observation that Algorithm~\ref{alg:sampling} selects $u_i$ only if $\bar\partial^{\hat{\tau}}_{u_i}F\big(\tx\wedge\onevect{\curly{u_1,u_2,\dotsc,u_{i-1}}}\big)\geq c\hat\tau/k$, we get
	\begin{align*}
	F(\tx)
	\geq{} &
	F(\onevect{\varnothing})
	+\sum_{i=j}^n\left[\tx_{u_i}\cdot\left(\bar\partial^{\hat{\tau}}_{u_i}F\big(\tx\wedge\onevect{\curly{u_1,u_2,\dotsc,u_{i-1}}}\big)
	-\hat\tau\cdot\frac{\varepsilon'}{20k}\right)\right]\\
	\geq{} &
	F(\onevect{\varnothing})
	+\hat\tau\cdot\left(\frac{c}{k}-\frac{\varepsilon'}{20k}\right)\cdot\sum_{i=j}^n\tx_{u_i}
	={} F(\onevect{\varnothing})
	+\hat\tau\cdot\left(\frac{c}{k}-\frac{\varepsilon'}{20k}\right)\cdot\onenorm{\tx}
	\geq{} \hat\tau\cdot\left(c-\frac{\varepsilon'}{20}\right)\enspace,
	\end{align*}
	where the last equality holds since $\tilde{x}$ contains a zero fraction of every element arriving before $u_j$ by the definition of $j$, and the last inequality follows since $F$ is non-negative and we assumed that $\onenorm{\tx}=k$.
\end{proof}

Next, we bound in the following lemma the value of $F(\tx)$ when $\onenorm{\tx}<k$. To prove this lemma we use the same steps used above to derive Corollary~\ref{cor:belowK}. Note that the bound we get here on $F(\tx)$ is equal up to a small error term to the bound in this corollary.
\begin{lemma}
\label{lem:apxF-belowK}
	Assuming $\cE$ holds, if $\onenorm{\tx}<k$, then $F(\tx)\geq
	(1-p)\cdot\Big[p\cdot f(OPT)+(1-p)\cdot f\big(OPT\setminus\supp(\tx)\big)\Big]-\hat\tau\cdot(c+\varepsilon'/20)$.
\end{lemma}
\begin{proof}
As done in the proofs of Lemmata~\ref{lem:belowK_lower} and~\ref{lem:belowK_upper} (which were used to prove Corollary~\ref{cor:belowK}), we find lower and upper bounds for $F(\tx+\onevect{OPT\setminus\supp(\tx)})$. Since Algorithms~\ref{alg:main_algorithm} and~\ref{alg:sampling} both handle the scenario of an arriving element getting past the threshold for acceptance similarly, Observation~\ref{obs:x_type} and Lemma~\ref{lem:belowK_lower} apply to $\tx$. Thus,
	\begin{equation}
	\label{eq:apxF-belowK_lower}
		F\left(\tx+\onevect{OPT\setminus\supp(\tx)}\right)
		\geq{} (1-p)\cdot\bracks{p\cdot f(OPT)+(1-p)\cdot f\big(OPT\setminus\supp(\tx)\big)}
		\enspace.
	\end{equation}
	Moreover, for every element $u_i\in OPT\setminus\supp(\tx)$ it holds that, if $i < j$, then $\partial_{u_i} F(\tx) \leq \partial_{u_i} F(\tx \wedge \onevect{u_1, u_2, \dotsc, u_{i - 1}}) < c\hat{\tau} /k$ by the definition of $j$, and if $i \geq j$, then
	\[
	\partial_{u_i}F(\tx)-\hat\tau\cdot\frac{\varepsilon'}{20k}
	\leq{} \partial_{u_i}F(\tx)-\frac{\varepsilon'(1-\varepsilon')}{20k}\cdot f(OPT)
	\leq{} \bar\partial^{\hat{\tau}}_{u_i}F(\tx)
	<{} \frac{c\hat{\tau}}{k} \enspace, \]
	where the first inequality follows from the definition of $\hat\tau$, the second inequality follows from the assumption that the event $\cE$ happens, and the last inequality follows since the elements in $OPT\setminus\supp(\tx)$ were rejected by Algorithm~\ref{alg:sampling} and $f$ is submodular.
	The submodularity of $f$ also implies
	\begin{align}
	\label{eq:apxF-belowK_upper}
		F\big(\tx+\onevect{OPT\setminus\supp(\tx)}\big)
			\leq{} F(\tx)+\mspace{-40mu}\sum_{u\in OPT\setminus\supp(\tx)}
			{\mspace{-40mu}\partial_uF(\tx)}
			\leq{} &
			F(\tx)+|OPT \setminus \supp(\tx)| \cdot \hat\tau\cdot \left(\frac{c}{k}+\frac{\varepsilon'}{20k}\right)\nonumber\\
			\leq{} &
			F(\tx)+\hat\tau\cdot\left(c+\frac{\varepsilon'}{20}\right)\enspace,
	\end{align}
	where the last inequality follows since $OPT$ is a feasible solution, and thus, contains at most $k$ elements. The lemma now follows by combining Inequalities~\eqref{eq:apxF-belowK_lower} and~\eqref{eq:apxF-belowK_upper}.
\end{proof}

The following lemma is obtained by combining the results of the previous two. This lemma corresponds to Lemma~\ref{lem:bound-solution} from the analysis of Algorithm~\ref{alg:main_algorithm}.

\begin{lemma}
\label{lem:apxF-bound-solution}
	Assuming $\cE$ holds,
	$\Exp{\max\curly{f(S_1^{\hat\tau}),f(S_2^{\hat\tau})}}
	\geq\hat\tau\cdot\bracks{
	\min\curly{c,\frac{\alpha(1-p-c)}{\alpha+(1-p)^2}}-\frac{\varepsilon'}{20}}$.
	In particular, for $c=\frac{\alpha(1-p)}{2\alpha+(1-p)^2}$ we get
	$\Exp{\max\curly{f(S_1^{\hat\tau}),f(S_2^{\hat\tau})}}\geq\hat\tau\cdot\bracks{
	\frac{\alpha(1-p)}{2\alpha+(1-p)^2}-\frac{\varepsilon'}{20}}$.
\end{lemma}
\begin{proof}
	If $\onenorm{\tx}=k$, then by the definition of $S_1^{\hat\tau}$
	\begin{equation}
	\label{eq:apxF-solution_exactlyK}
		\Exp{\max\curly{f(S_1^{\hat\tau}),f(S_2^{\hat\tau})}}
		\geq{} \Exp{f(S_1^{\hat\tau})}\geq{} F(\tx)
		\geq{} \hat\tau\cdot\left(c-\frac{\varepsilon'}{20}\right)\enspace,
	\end{equation}
	where the last inequality follows from Lemma~\ref{lem:apxF-exactlyK}.
	To address the case in which $\onenorm{\tx}<k$, note that $OPT\cap\supp(\tx)$ is a feasible solution within the support of $\tx$. Thus, $\Exp{f(S_2^{\hat\tau})}\geq\alpha\cdot f(OPT\cap\supp(\tx))$ by the definition of $S_2^{\hat\tau}$. Therefore, since no convex combination of two values is higher than their maximum, Lemma~\ref{lem:apxF-belowK} yields
	\begin{align*}
		\Exp{\max\curly{f(S_1),f(S_2)}}\mspace{-81mu}&\mspace{81mu}
		\geq
		\max\curly{\Exp{f(S_1)},\Exp{f(S_2)}}\\
		\geq{}&
		\max\bigg\{(1-p)\cdot\big[p\cdot f(OPT)+(1-p)\cdot f(OPT\setminus\supp(\tx))\big]-\hat{\tau}\left(c+\frac{\varepsilon'}{20}\right),\\
		&\alpha\cdot{f(OPT\cap\supp(\tx))}\bigg\}\\
		\geq{} &
		\frac{\alpha}{\alpha+(1-p)^2}\cdot
		\bracks{(1-p)\cdot\bracks{p\cdot f(OPT)+(1-p)\cdot f(OPT\setminus\supp(\tx))}-\hat{\tau}\left(c+\frac{\varepsilon'}{20}\right)}\\
		&+\frac{(1-p)^2}{\alpha+(1-p)^2}\cdot\alpha\cdot f(OPT\cap\supp(\tx))
		\enspace.
	\end{align*}
	Observe that the rightmost side in the above inequality is equal to
	\[
	\frac{\alpha(1-p)^2}{\alpha+(1-p)^2}
	\cdot\bracks{f(OPT\cap\supp(\tx))+f(OPT\setminus\supp(\tx))}
	+\frac{\alpha\bracks{p(1-p)\cdot f(OPT)-\hat{\tau}\left(c+\frac{\varepsilon'}{20}\right)}}{\alpha+(1-p)^2}\enspace.
	\]
	Moreover, since $f$ is submodular and non-negative, it holds that $f(OPT\cap\supp(\tx))+f(OPT\setminus\supp(\tx))\geq f(OPT)\geq\hat\tau$.
	Therefore, by combining all the above, we get
	\begin{align}
	\label{eq:apxF-solution_belowK}
		\Exp{\max\curly{f(S_1),f(S_2)}}
		\geq{} &
		\frac{\alpha(1-p)^2}{\alpha+(1-p)^2}\cdot\hat\tau +\frac{\alpha\bracks{p(1-p)\cdot \hat\tau-\left(c+\frac{\varepsilon'}{20}\right)\hat\tau}}{\alpha+(1-p)^2}\nonumber\\
		={} &
		\hat\tau\cdot\frac{\alpha\left(1-p-c-\frac{\varepsilon'}{20}\right)}{\alpha+(1-p)^2}
		\geq{}
		\hat\tau\cdot\bracks{\frac{\alpha(1-p-c)}{\alpha+(1-p)^2}-\frac{\varepsilon'}{20}}\enspace.
	\end{align}
	The first part of the lemma now follows from Inequalities~\eqref{eq:apxF-solution_exactlyK} and~\eqref{eq:apxF-solution_belowK}. To prove the second part, note that in the proof of Lemma~\ref{lem:bound-solution} we have shown that setting $c=\frac{\alpha(1-p)}{2\alpha+(1-p)^2}$ in $\frac{\alpha(1-p-c)}{\alpha+(1-p)^2}$ yields $c=\frac{\alpha(1-p-c)}{\alpha+(1-p)^2}$. Hence, when the parameter $c$ is set as above, 
	\[c-\frac{\varepsilon'}{20}=\frac{\alpha(1-p-c)}{\alpha+(1-p)^2}-\frac{\varepsilon'}{20}=\frac{\alpha(1-p)}{2\alpha+(1-p)^2}-\frac{\varepsilon'}{20}\enspace. \qedhere\]
\end{proof}

We are now ready to prove Theorem~\ref{trm:polynomial}.
\begin{proof}[Proof of Theorem~\ref{trm:polynomial}]
In the proof of Theorem~\ref{trm:polynomial-knowing-tau-oracle}, we have discussed an algorithm by Buchbinder et al.~\cite{DBLP:conf/soda/BuchbinderFNS14} for the offline version of our problem, and described the approximation ratio of this algorithm for instances in which the size of the ground set is upper bounded by $\lceil k/p \rceil$. In particular, we showed that for $p=0.24$ this algorithm can find a solution whose expected value is at least a $0.460675$ fraction of the optimal value. Thus, it is possible to implement Algorithm~\ref{alg:sampling} with $\alpha=0.460675$, $p=0.24$, $c=\frac{\alpha(1-p)}{2\alpha+(1-p)^2}$ and $\varepsilon'=10^{-4}$. By Lemma~\ref{lem:apxF-bound-solution}, conditioned on the event $\cE$, the expected value of the output set $T$ of such an implementation is at least
		\begin{align*}
		\Exp{f(T)}
		\geq
		\hat\tau\cdot\bracks{\frac{\alpha(1-p)}{2\alpha+(1-p)^2}-\frac{\varepsilon'}{20}}
		\geq{} &
		0.9999\cdot f(OPT)\cdot\bracks{\frac{0.460675\cdot(1-0.24)}{2\cdot0.460675+(1-0.24)^2}-0.000005}\\
		\geq{} &
		f(OPT)\cdot\bracks{\frac{0.350113}{1.49895}-0.000005}
		\geq 0.233567\cdot f(OPT)\enspace.
	\end{align*}
Using the low of total expectation, we now get
	\begin{align*}
		\Exp{f(T)}
		={} &
		\pr{\cE}\cdot \Exp{f(T) \mid \cE} + \Pr[\bar{\cE}] \cdot \Exp{f(T) \mid \bar{\cE}}
		\geq
		\pr{\cE}\cdot \Exp{f(T) \mid \cE}\\
		\geq{} &
		\left(1 - \frac{\varepsilon'}{20}\right) \cdot \Exp{f(T) \mid \cE}
		\geq
		0.999995 \cdot 0.233567\cdot f(OPT)
		\geq
		\frac{f(OPT)}{4.282}
		\enspace,
	\end{align*}
	where the first inequality holds by the non-negativity of $f$ and the second inequality follows from Corollary~\ref{cor:apxF-prob-all-i}. 
This shows that the above mentioned implementation of Algorithm~\ref{alg:sampling} achieves the approximation guarantee of Theorem~\ref{trm:polynomial}.

To complete the proof of the theorem, it remains to observe that, for the above specified values for the parameters $p$, $\varepsilon'$ and $c$, Lemma~\ref{lem:apxF-space} shows that Algorithm~\ref{alg:sampling} obeys the space complexity guarantees of Theorem~\ref{trm:polynomial} because the algorithm for calculating $S^\tau_2$ requires only $\poly(1/\varepsilon') \cdot \tilde{O}(k)$ space (as explained in the proof of Theorem~\ref{trm:polynomial-knowing-tau-oracle}).
\end{proof}

\end{document}